\documentclass{article}
\def\DEBUG{1}
\usepackage[preprint, nonatbib]{neurips_2020}
\usepackage[utf8]{inputenc}

\usepackage{amsmath, amsfonts,amssymb, amsthm}
\usepackage{thm-restate}

\usepackage{xspace}

\usepackage{graphicx}
\graphicspath{{figs/}}
\usepackage{subfig}

\usepackage{enumitem}
\setlist{topsep=0pt,itemsep=0pt,leftmargin=*}

\usepackage[english]{babel}
\usepackage{xcolor}

\allowdisplaybreaks

\usepackage{algorithm}
\usepackage{algorithmicx}
\usepackage{algpseudocode}
\algtext*{EndWhile}
\algtext*{EndIf}
\algtext*{EndFor}
\algtext*{EndFunction}
\algtext*{EndProcedure}

\newtheorem{theorem}{Theorem}
\newtheorem{definition}[theorem]{Definition}
\newtheorem{lemma}[theorem]{Lemma}

\newcommand{\floor}[1]{{\left\lfloor#1\right\rfloor}}

\newcommand{\R}{{\mathbb{R}}}

\mathchardef\mhyphen="2D
\newcommand{\kmo}{\ensuremath{\mathsf{k\mhyphen Med\mhyphen O}}\xspace}

\newcommand{\cost}{\mathsf{cost}}
\newcommand{\poly}{\mathrm{poly}}

\newcommand{\opt}{\mathsf{opt}}

\newcommand{\bfp}{\mathbf{p}}

\DeclareMathOperator*\union{\bigcup}

\ifdefined\DEBUG
	\newcommand{\sh}[1]{{\color{red} #1}}

	\def\rem#1{{\marginpar{\raggedright\scriptsize #1}}}
	\newcommand{\shr}[1]{\rem{\small\textcolor{red}{$\bullet${\tiny Shi: #1\\}}}}
	\newcommand{\janr}[1]{\rem{\small\textcolor{brown}{$\bullet${\tiny Janard: #1\\}}}}
	\newcommand{\xguor}[1]{\rem{\small\textcolor{blue}{$\bullet${\tiny Xiangyu: #1\\}}}}

	\newcommand{\remove}[1]{{\color{lightgray} #1}}
	
\else
	\newcommand{\sh}[1]{#1}

	\newcommand{\shr}[1]{}
	\newcommand{\janr}[1]{}
	\newcommand{\xguor}[1]{}
	\newcommand{\jaiyir}[1]{}
	
	\newcommand{\remove}[1]{}
	
\fi

\makeatletter
\newcommand{\settitle}{\@maketitle}
\makeatother

\begin{document}


\title{Consistent $k$-Median: Simpler, Better and Robust}

\author{ 
		Xiangyu Guo \thanks{Department of Computer Science and Engineering, University at Buffalo, {\tt xiangyug@buffalo.edu }} \and
		Janardhan Kulkarni \thanks{The Algorithms Group, Microsoft Research, Redmond, {\tt jakul@microsoft.com}} \and
		Shi Li \thanks{Department of Computer Science and Engineering, University at Buffalo, {\tt shil@buffalo.edu}} \and
		Jiayi Xian \thanks{Department of Computer Science and Engineering, University at Buffalo, {\tt jxian@buffalo.edu}}
	   }

\date{}

\maketitle
	\begin{abstract}
	In this paper we introduce and study the online consistent $k$-clustering with outliers problem, generalizing the non-outlier version of the problem studied in Lattanzi-Vassilvitskii \cite{LattanziV17}. 
	We show that a simple local-search based online algorithm can give a bicriteria constant approximation for the problem with $O(k^2 \log^2 (nD))$ swaps of medians (recourse)  in total, where $D$ is the diameter of the metric. When restricted to the problem without outliers, our algorithm is simpler, deterministic and gives better approximation ratio and recourse, compared to that of Lattanzi-Vassilvitskii \cite{LattanziV17}.
	
%
\end{abstract}	
\section{Introduction}
\label{sec:intro}

Clustering is one of the most fundamental primitives in unsupervised machine learning, and  $k$-median clustering is one of the most widely used primitives in practice.
Input to the problem consists of a set $C$ of $n$ points, a set $F$ of potential median locations, a metric space $d: (C \cup F) \times (C \cup F)  \rightarrow \R_{\geq0}$. The goal is to choose a subset $S \subseteq F$ of cardinality at most $k$ so as to minimize $\sum_{j \in C} d(j, S)$ where $d(j,S) := \min_{i \in S} d(j,i)$ is the distance from $j$ to its nearest chosen median. The problem is known to be NP-hard and several constant factor approximation algorithms are known to the problem \cite{CGST99, JV99, AGKMMP01, LS13, BPRST17}. 

In many real world applications, the set of data points arrive over time in an {\em online} fashion. 
For example, images, videos, documents get added over time, and clustering algorithms in such applications need to assign a label (or a median) to each newly added point in an online fashion.
A natural framework to study these online clustering problems is using {\em competitive analysis}, where the goal is to assign each arriving data point  {\em irrevocably} to an existing cluster or start a new cluster containing the point. 
Unfortunately, the competitive analysis framework is too strong, and it is provably impossible to maintain a good quality clustering of data points if one insists on the irrevocable decisions \cite{liberty2016algorithm}.
Recently, Lattanzi and Vassilvitskii \cite{LattanziV17} observed that in many applications the decisions need not be irrevocable, however the online algorithm should not do too many {\em re-clustering} operations.
Motivated by such settings they initiated the study of {\em consistent $k$-clustering} problem.
The goal in consistent $k$-clustering is twofold:  
\begin{itemize}
\item \textbf{Quality}: Guarantee at all the times that we have a clustering of the points that is a good approximation to  the optimum one.
\item \textbf{Consistency:} The chosen medians should be stable and not change too frequently over the sequence of data point insertions.
\end{itemize}

Lattanzi and Vassilvitskii \cite{LattanziV17} measured the number of changes to the set of chosen medians using the notion of \emph{recourse} -- a concept also studied in online algorithms \cite{Gupta015, GuptaKS14, BernsteinHR19}. 
The total recourse of an online algorithm is defined as the number of changes it makes to the solution. 
Specially for the $k$-median problem, if $S_t$ corresponds to the set of chosen medians at time $t$ and $S_{t+1}$ at time $t+1$, then the recourse at time step $t+1$ is $|S_{t+1} \setminus S_t |$. \footnote{One can also define the recourse as $|S_{t+1} \setminus S_t| + |S_t \setminus S_{t+1}|$, but if we assume $|S_t| = |S_{t+1}| = k$, this is exactly $2\cdot|S_{t+1} \setminus S_t |$.}
The total recourse of an online algorithm is the sum of recourse across all the time steps.
An online algorithm with small recourse ensures that the chosen medians do not change too frequently and hence is consistent.
In particular, it forbids an algorithm from simply recomputing the solution from scratch at each time step.
This is a very desirable property of a clustering algorithm in applications, as we do not want to change the label assigned to data points (which corresponds to cluster centers) as the data set keeps growing. 
Broadly speaking, recourse is also a measure of {\em stability} of an online algorithm.
Lattanzi and Vassilvitskii \cite{LattanziV17} showed that one can maintain an $O(1)$ approximation to the $k$-median problem with $O(k^2 \log ^4 n)$ total recourse.
More recently, Cohen-Addad {\em et al} \cite{Cohen-AddadHPSS19} studied facility location (and clustering problems) from the 
perspective of both dynamic  and consistent clustering frameworks. 
See related work section for more details. 

A drawback of using $k$-median clustering on real-world data sets is that it is not robust to noisy data, i.e., a few outliers can completely change the cost as well as structure of solutions.  
Recognizing this shortcoming, Charikar et al.  \cite{CKMN01} introduced a {\em robust} version of $k$-median problem called {\em $k$-median with outliers}. 
The problem is similar to $k$-median problem except one crucial difference:
An algorithm for $k$-median with outliers does not need to cluster all the points but can choose to ignore a small fraction of the input points.
The number of points an algorithm can ignore is given as a part of the input, and is typically set to be a small fraction of the overall input.  

Formally, in the $k$-median with outliers (\kmo) problem, we are given $F$, $C$, $d$ and $k$ as in the $k$-median problem. Additionally, we are given an integer $z \leq n = |C|$. The goal is to choose a set $S \subseteq F$ of $k$ medians, so as to minimize 
\begin{align*}
	\textstyle \min_{O \subseteq C: |O| = z} \sum_{j \in C \setminus O} d(j, S).
\end{align*}
 The set $O$ of points are called \emph{outliers} and are not counted in the cost of the solution $S$. Thus the parameter $z$ specifies the number of outliers.  
Notice that when $S$ is given, the set $O$ that minimizes $\sum_{j \in C \setminus O} d(j, S)$ can be computed easily: It contains the $z$ points $j \in C$ with the largest $d(j, S)$ value.  Therefore for convenience we shall simply use a set $S \subseteq F$ of size $k$ to denote a solution to a \kmo instance. The \kmo problem is not only a more robust objective function but also helps in removing outliers -- a very important issue in the real world datasets \cite{OPRC14, ChawlaGionis13}. 
In fact such a joint view of clustering and outlier elimination has been observed to be more effective, and has attracted significant attention both in theory and practice \cite{Chen08, ChawlaGionis13, GKLMV17, rujeerapaiboon2019size, KLS18}. 

In this paper, we study the \kmo problem in the online {\em consistent $k$-clustering} framework of  Lattanzi and Vassilvitskii.
The goal is to maintain a good quality (approximate) solution to the problem at all times while minimizing the total recourse of the online algorithm. (The total recourse is still defined as $\sum_{t} |S_t \setminus S_{t-1}|$.)
Though $O(1)$-approximation algorithms for \kmo are known in the offline setting \cite{Chen08, KLS18},  it seems hard to extend these algorithms to the online setting. Instead, we resort to \emph{bicrtieria approximate solutions} for the \kmo problem:
\begin{definition}
	We say a solution $S \subseteq F$ of $k$ medians is a $(\beta, \alpha)$-bicriteria approximation to the $k$-median with outliers instance $(F, C, d, k, z)$ for some $\alpha, \beta \geq 1$, if there exists a set $O \subseteq C$ of size at most $\beta z$ such that $\sum_{j \in C \setminus O} d(j, S) \leq \alpha \cdot \opt$, where $\opt$ is the cost of the optimum solution for the instance with $z$ outliers. 
\end{definition}
So, a $(\beta, \alpha)$-approximate solution removes at most $\beta z$ outliers and has cost at most $\alpha$ times the cost of the optimum solution with $z$ outliers.  

\textbf{Online Model for $k$-Median with Outliers}\ \  We now describe the online model for the \kmo problem. Recall that a \kmo instance is given by $F, C, d, k$ and $z$. As in \cite{LattanziV17}, we assume $k$ is given at the beginning of the algorithm, and $C$ and $d$ will be given online.  We use $n$ to denote the total number of  clients that will arrive.  

Depending how $F$ is given, we have two slightly different online settings:
\begin{itemize}
	\item In the \emph{static $F$} setting,  we assume $F$ is independent of  $C$ and is given at the beginning of the online algorithm. In each time step, one point in $C$ arrives and its distances to $F$ are revealed. \footnote{It is easy to see that in the \kmo problem, only distances between $F$ and $C$ are relevant.} 
	\item In the $F = C$ setting, we assume we always have $F = C$.  Whenever a point arrives, its distances to previously arrived points are revealed, and  the point is then added to both $C$ and $F$.  
\end{itemize}
 The $F = C$ setting is more natural for clustering applications and is the one used in \cite{LattanziV17}.   On the other hand, the static $F$ setting arises in applications where we want to build $k$ facilities to serve a set $C$ of clients that arrive one by one. In these applications, the set $F$ of potential locations to build facilities is independent of $C$ and often does not change over time.   The analysis of our algorithm works directly for the static $F$ setting, but needs a small twisting in the $F = C$ setting. 
 
It remains to describe how $z$ is given.  For simplicity, we assume $z$ is fixed and given at the beginning of the algorithm; we call this the static $z$ setting. In a typical application, $z$ may increase as more and more points arrive, and we call this setting the incremental $z$ setting. 
We can reduce the incremental $z$ setting to the static $z$ setting in the following way. We maintain an integer $z' \in [z, (1+\epsilon) z)$ and use $z'$ as the given number of outliers.  This will incur a factor of $(1+\epsilon)$ in the first factor of the bicriteria approximation.  During our algorithm, whenever $z$ becomes more than $z'$, we update $z'$ to $\floor{(1+\epsilon)z}$. We define an \emph{epoch} to be a maximal period of time steps with the same $z'$ value. So within an epoch, $z'$ value does not change.  The number of epochs is at most $O(\log_{1+\epsilon} n) = O\Big(\frac{\log n}{\epsilon}\Big)$. Thus, if we have an online $(\beta, \alpha)$-approximation algorithm for $\kmo$ with total recourse $R$ in the static $z$ setting, we can obtain an $((1+\epsilon)\beta, \alpha)$-approximation algorithm with total recourse $O\Big(\frac{R\log n}{\epsilon}\Big)$ in the incremental $z$ setting.  Thus throughout the paper, we only focus on the static $z$ setting, that is, $z$ is fixed and given at the beginning of the algorithm.

\textbf{Our Results}\ \  The main contribution of the paper is the following. Recall that $n$ is the total number of points that will arrive during the whole algorithm. We assume all distances are integers and define $D$ to be the diameter of the metric $d$.
\begin{theorem}
\label{thm:mainoutlier}
	There is a deterministic $(O(1), O(1))$-bicriteria approximation algorithm for the online $k$-median with outliers problem with a total recourse of $O\big(k^2\log n \log (nD)\big)$.
\end{theorem}
When restricted to the case without outliers (i.e, $z = 0$), our algorithm gives the following.
\begin{theorem}
\label{thm:kmedian}
There is a deterministic $O(1)$-approximation algorithm to the consistent $k$-median problem with $O\left(k^2 \log n \log (nD)\right)$ total recourse.
\end{theorem} 

The recourse achieved by our algorithm is $O(\log ^2 n)$ factor better than the result of Lattanzi and Vassilvitskii \cite{LattanziV17}. \footnote{In \cite{LattanziV17}, it is assumed that $D = \poly(n)$ and thus $O(\log(nD)) = O(\log n)$.}
They also showed a lowerbound of $\Omega(k \log n)$ on the total recourse, hence our result also takes a step towards achieving the optimal recourse for this basic problem.  

Lemma \ref{lemma:okm-main} that appears later gives a formal statement of the guarantees obtained by our algorithm.
In Lemma \ref{lemma:okm-main} we prove a more general result, where one can trade-off running time and the approximation factor achieved by our algorithm by fine-tuning certain parameters.  In particular, by appropriate tuning of parameters we can achieve $3+\epsilon$ approximation in time $n^{O(1/\epsilon)}$, matching the approximation factor achieved by local search algorithm in the offline setting, and also improves the unspecified $O(1)$ factor achieved by \cite{LattanziV17}.
Finally, our algorithm is deterministic while that of \cite{LattanziV17} is randomized and only succeeds with high probability.

\textbf{Our Techniques}\ \ Unlike many of the previous results on the online $k$-median problem and the related facility location problem, which are based on Meyerson's sampling procedure \cite{Meyerson2001}, our approach is based on {\em local search}. When restricted to the $k$-median without outliers problem, at every time step, it repeatedly applies \emph{$\rho$-efficient} swap operations until no such operations exist: These are the swaps that can greatly decrease the cost of the solution (See Definition~\ref{def:efficient-op}). Via standard analysis, one can show that this gives an $O(1)$-approximation for the problem. To analyze the total recourse of the algorithm, we establish a crucial lemma that the total cost increment due to the arrival of clients is small.  Compared to Meyerson's sampling technique, local search has two advantages:  (i) The approximation ratio can be made to be $3+\epsilon$, which matches the best offline approximation ratio for $k$-median based on local search. (ii) Local-search based algorithms are deterministic in general.

Very recently, similar techniques were used in \cite{GKLX20} to derive online algorithms for the related facility location problem.  We extend their ideas to the $k$-median problem, and more importantly, the $k$-median with outliers problem. 

One barrier to extend the algorithm to the outlier setting is that the analysis for the local search algorithm breaks down if we impose the constraint that the number of outliers can be at most $z$. To circumvent the barrier, we handle the constraint in a soft manner: We introduce a penalty cost $p$, and instead of requiring the number of outliers to be at most $z$, we pay a cost of $p$ for every outlier in the solution. By setting $p$ appropriately, we can ensure that the algorithm does not produce too many outliers, while at the same time maintaining the $O(1)$ approximation ratio.   Indeed, in the offline setting, our algorithm gives the first $(O(1), O(1))$-bicritiera approximation for \kmo based on local search. Prior to our work, in the offline setting, Gupta et al \cite{GKLMV17} developed a bicriteria approximation for the problem, but it needs to violate the outlier constraint by a factor of $O(k\log (nD))$. On the other hand, though $O(1)$-approximation algorithms for \kmo were developed in \cite{Chen08} and \cite{KLS18}, unlike our local search based algorithm, they are hard to extend to the online setting.

\textbf{Other Clustering Objectives}\ \  We remark that our algorithm and analysis can be easily extended to the $k$-means objective, and more generally, the sum of $q$-th power of distances for any constant $q \geq 1$. However for the cleanness of presentation, we choose to only focus on the $k$-median objective.  

\textbf{Related work}\ \  As we mentioned earlier, Cohen-Addad {\em et al} \cite{Cohen-AddadHPSS19} studied facility location and clustering problems from the perspective of both dynamic  and consistent clustering frameworks. 
In the dynamic setting, data points are both added and deleted, and the emphasis is to maintain good quality solutions while minimizing the  {\em time} it takes to update the solutions.
For the facility location problem, they gave an $O(1)$ approximation algorithm with almost optimal $O(n)$ {\em total} recourse and $O(n \log n)$ {\em per step} update time. 
They also extended their algorithm for facility location to the $k$-median and $k$-means problems (without outliers), achieving a constant factor approximate solution with $\tilde{O}(n + k^2)$ {\em per step} update time.
Unfortunately, they do not state the total recourse of their algorithms.
To our understanding, the total recourse of their algorithms can be as large as $O(n)$. 
However, they also consider a {\em harder} setting where data points are being both inserted and deleted.
We believe that finding a consistent $k$-clustering algorithm, where the emphasis is more on the stability of cluster centers than the update time, for the case when data points are inserted and deleted is an important open problem. 

For more details regarding clustering problems in the context of dynamic and online algorithms, we refer the readers to \cite{Charikar1997, Meyerson2001, Fotakis2011, Goranci19, GKLX20} and references therein.

{\em All the omitted proofs are given in the supplementary material.}

	\section{An Offline Local Search Algorithm for $k$-Median with Outliers}
\label{sec:okm}

In this section, we describe an offline local search algorithm for \kmo that achieves an $(O(1), O(1))$-bicriteria-approximation ratio.   To allow trade-offs among the approximation ratio, number of outliers and running time, we introduce two parameters: an integer $\ell \geq 1$ and a real number $\gamma > 0$.  The algorithm gives $\Big(\big(1+\frac1\ell)(1+\gamma), (3+\frac2\ell)\big(1+\frac1\gamma\big)\Big)$-bicriteria approximation in $n^{O(\ell)}$ time.   In particular, we can set $\ell =\gamma= \Theta(1/\epsilon)$  to get an approximation ratio  $3+\epsilon$ with $O\Big(\frac{z}{\epsilon}\Big)$ outliers and $n^{O(1/\epsilon)}$-time,  matching the best approximation ratio for $k$-median based on local search. To obtain any $(O(1), O(1))$-bicriteria approximation, it suffices and is convenient to set $\ell = \gamma = 1$.  This offline algorithm will serve as the baseline for our online algorithm for \kmo. 

The main idea behind the algorithm is that we convert the problem into the $k$-median with \emph{penalty} problem.  Compared to \kmo, in the problem we are not given the number $z$ of outliers, but instead we are given a penalty cost $p \geq 0$ for not connecting a point. Our goal is to choose $k$ medians and connect some points to the $k$ medians so as to minimize the sum of the connection cost and penalty cost.  So, we shall use the parameter $p$ to control the number of outliers in a soft way. 

Indeed, the $k$-median with penalty problem is equivalent to the original $k$-median problem up to the modification of the metric. For every two points $u, v \in F \cup C$, we define $d_p(u,v) := \min\{d(u, v), p\}$. 
%
Then it is easy to see that, the $k$-median with penalty problem becomes the $k$-median problem on the metric $d_p$.  For a set $S \subseteq F$ of $k$ medians, we define  $\cost_{p}(S) :=  \sum_{j \in C} d_p(j, S)$
to be the cost of the solution $S$ to the $k$-median instance with metric $d_p$, or equivalently, the $k$-median instance on metric $d$ with per-outlier penalty cost $p$. 

\textbf{Swap Operations for $k$-Median with Outliers}\ \  Given a set $S \subseteq F$ of $k$ medians, and an integer $\ell \geq 1$, an $\ell$-swap on $S$ is a pair $(A^*, A)$ of medians, such that $A \subseteq S, A^* \subseteq F \setminus S$ and $|A| = |A^*| \leq \ell$.  Applying the swap operation $(A^*, A)$ on $S$ will update $S$ to $S \cup A^* \setminus A$. Notice that after the operation $S$ still has size $k$.  We simply say $(A^*, A)$ is a swap on $S$ if it is an $\ell$-swap for some $\ell \geq 1$.

\begin{definition}[Efficient swaps]\label{def:efficient-op}
For any $\rho, p \geq 0$, a swap $(A^*, A)$ on a solution $S \subseteq F, |S| =k$ is said to be $\rho$-efficient w.r.t the penalty cost $p$, if we have $\cost_p(S \cup A^* \setminus A) <  \cost_{p}(S)  - |A|\rho$.
\end{definition}

In particular a $0$-efficient swap with respect to some penalty cost $p \geq 0$ is a swap whose application on $S$ will strictly decrease $\cost_p(S)$. The efficiency parameter $\rho$ will be used later in the online algorithm, in which we apply a swap only if it can decrease $\cost_p(S)$ significantly to guarantee that the recourse of our algorithm is small. 

The following theorem can be shown by modifying the analysis for the classic $(3+\frac2\ell)$-approximation local search algorithm for $k$-median \cite{Williamson}. We leave its proof to the supplementary material. 

\begin{restatable}{theorem}{adjuflapprox}\label{thm:local-search-apx-ratio}
	Let $S$ and $S^*$ be two sets of medians with $|S| = |S^*| = k$. Let $p, \rho \geq 0$, and $\ell \geq 1$ is an integer. If there are no $\rho$-efficient $\ell$-swaps on $S$ w.r.t the penalty cost $p$, then we have
		\begin{align*}
			\cost_p(S) \leq  \sum_{j \in C} \min\left\{\left(3 + \frac2\ell\right) d_p(j, S^*), \left(1+\frac1\ell\right)p\right\}  + k\rho.
		\end{align*}
\end{restatable}
To understand the theorem, we first assume $\rho = 0$; thus $S$ is a local optimum for the $k$-median instance defined by the metric $d_p$.  If we replace $\min\big\{\big(3+\frac2\ell\big) d_p(j, S^*), \big(1+\frac1\ell\big)p\big\}$ by $\big(3+\frac2\ell\big)d_p(j, S^*)$, then the theorem says that a local optimum solution for $k$-median is a $\big(3+\frac2\ell\big)$-approximation, which is exactly the locality gap theorem for $k$-median.  Using that $d_p$ has diameter $p$, we can obtain the improvement as stated in the theorem; this will be used to give a better trade-off between the two factors in the bicriteria approximation ratio. When $\rho \geq 0$, we lose an additive factor of $k\rho$ on the right side of the inequality. 

Theorem~\ref{thm:local-search-apx-ratio} immediately gives a $\left(\big(1+\frac1\ell\big)(1+\gamma), \big(3+\frac2\ell\big)\big(1+\frac{1}{\gamma}\big)\right)$-bicriteria approximation algorithm for the \kmo problem, for any $\gamma > 0$.  By binary search, we assume we know the optimum value $\opt$ for the \kmo instance. Let $p = \frac{(3\ell+2)\opt}{(\ell+1)\gamma z}$. Then we start from an arbitrary set $S$ of $k$ medians, and repeatedly apply $0$-efficient $\ell$-swaps w.r.t penalty cost $p$ on $S$, until no such swaps can be found. The running time of the algorithm is $n^{O(\ell)}$.\footnote{When the distances are not polynomially bounded, the running time of the algorithm may be large; but using an appropriate $\rho$ we can reduce the running time to polynomial by losing a factor of $(1+\epsilon)$ in the approximation ratio.}  Applying Theorem~\ref{thm:local-search-apx-ratio} with $S^*$ being the optimum  solution for the \kmo instance, we have that the final solution $S$ has $\cost_p(S) \leq \sum_{j \in C} \min\Big\{\big(3+\frac2\ell\big)d_p(j, S^*), \big(1+\frac1\ell\big)p\Big\} \leq \big(3+\frac2\ell\big)\opt + \big(1+\frac1\ell\big)zp$. 
The second inequality holds since for inliers $j$ in the solution $S^*$, we have $d_p(j, S^*) \leq d(j, S^*)$ and for outliers $j$ we have $d_p(j, S^*) \leq p$. 
We return $S$ as the set of medians, and let $j$ be an outlier if $d_p(j, S) = p$.  Then, the number of outliers our algorithm produces is at most $\big(1+\frac1\ell\big)z + \frac{\big(3+\frac2\ell\big)\opt}{p}= \big(1+\frac1\ell\big)z +  \big(1+\frac1\ell\big)\gamma z = \big(1+\frac1\ell\big)(1+\gamma)z$. The cost of the solution is at most $\big(3+\frac2\ell\big)\opt + \big(1+\frac1\ell\big)zp = \big(3+\frac2\ell\big)\opt + \left(3+\frac2\ell\right)\frac{\opt}{\gamma} = \big(3+\frac2\ell\big)\opt\left(1+\frac{1}{\gamma}\right)$.

	\section{Online Algorithm for $k$-Median with Outliers}
In this section, we give our online algorithm for \kmo that proves Theorem ~\ref{thm:mainoutlier} (and thus Theorem~\ref{thm:kmedian}). As mentioned earlier, indeed we give a more general result that allows trade-offs between the approximation ratio, the number of outliers and running time:
\begin{lemma}
	\label{lemma:okm-main}
	Let $\ell \geq 1$ be an integer, $\epsilon > 0$ be small enough and $\gamma > 0$ be a real number. There is a deterministic $n^{O(\ell)}$-time algorithm for online $k$-median with outliers with a total recourse of $O\big(\frac{k^2\log n \log (nD)}{\epsilon}\big)$.  The algorithm achieves a bicriteria approximation of  $\Big(\frac{1}{1-\epsilon}\big(1+\frac1\ell\big)(1+\gamma), \frac{1}{1-\epsilon}\big(3+\frac2\ell\big)(1+\frac{2}{\gamma}\big)\Big)$ in the static $F$ setting, and $\Big(\frac{1}{1-\epsilon}\big(1+\frac1\ell\big)(1+\gamma), \frac{1}{1-\epsilon}\big(3+\frac2\ell\big)(1+\frac{4}{\gamma}\big)\Big)$ in the $F = C$ setting.
\end{lemma}
By setting $\ell = \gamma = 1$ and $\epsilon$ to be  a small enough constant, Lemma~\ref{lemma:okm-main} implies Theorem~\ref{thm:mainoutlier}.  On the other hand, one can set $\ell = \gamma = \frac1\epsilon$ to achieve an approximation ratio of $3 + O(\epsilon)$ with $O(\frac{z}{\epsilon})$ outliers and running time $n^{O(1/\epsilon)}$.
The goal of this section is to prove Lemma~\ref{lemma:okm-main}.  To explain our main ideas more clearly,  we assume $F$ is static: the set $F$ of potential medians is fixed and given at the beginning of the online algorithm.  
In the supplementary material Section~\ref{subsec:F=C}, we show how the algorithm can be extended to the setting where $F = C$.  

To avoid the case where the optimum solution has cost $0$, we add an additive factor of $0.1$ in all definitions of costs: the cost of a solution to a \kmo instance, and $\cost_p(S)$. We can think of that in the instance we have one point and one median that are $0.1$ distance apart and have distance $\infty$ to all other points in the metric.   Since all distances are integers and the approximation ratio we are aiming at is less than 10, the additive factor of $0.1$ does not change our approximation ratio. Theorem~\ref{thm:local-search-apx-ratio} holds with an additive factor of $0.1$ added to the right side of the inequality. 


In essence, our algorithm repeatedly applies $\rho$-efficient swaps w.r.t the penalty cost $p$, for some carefully maintained parameters $\rho$ and $p$.  The main algorithm is described in Algorithm~\ref{alg:online-k-median}. In each time $t$, we add the arrival point $j_t$ to $C$ (Step \ref{step:include-j-t}).  Then we repeatedly perform $\big(\rho:=\frac{\epsilon\cdot\cost_{p}(S)}{k}\big)$-efficient swaps until no such operation exists (loop~\ref{step:okm-while}). If the solution $S$ obtained has more than $\frac1{1-\epsilon}\big(1+\frac1\ell\big)(1+\gamma)z$ outliers (defined as the points $j$ with $d_p(j, S) = p$, or equivalently $d(j, S) \geq p$), we then double $p$ (Step~\ref{step:increase-p}) and redo the while loop.  At the beginning of the algorithm, we set $p$ to be a small enough number (Step~\ref{step:init-z}).


\begin{algorithm}[H]
	\caption{Online algorithm for $k$-median} \label{alg:online-k-median}
	\begin{algorithmic}[1]
		\State $p \gets \min\big\{\frac1{10\gamma z}, 0.1\big\}$ \label{step:init-z}
		\For {$t \gets 1$ to $n$}
			\State $C \gets C \cup \{j_t\}$ \label{step:include-j-t}
				\While{there exists a $\left(\rho:= \frac{\epsilon\cdot \cost_{p}(S)}{k}\right)$-efficient $\ell$-swap on $S$ w.r.t the penaty cost $p$}  \label{step:okm-while}
					\State perform the swap operation
				\EndWhile
				\If {$d_p(j, S) = p$ for more than $\frac1{1-\epsilon}\big(1+\frac1\ell\big)(1+\gamma)z$ points $j \in C$} \label{step:check-outlier}
					\State $p \gets 2p$ \label{step:increase-p}
					\State \textbf{goto} \ref{step:okm-while} 
				\EndIf
		\EndFor
	\end{algorithmic}
\end{algorithm}


\subsection{Approximation Ratio of the Algorithm} \label{subsec:BoundingApprox}
We start from analyzing the approximation ratio of the algorithm.  At any moment of the algorithm, we use $\opt$ to denote the cost of the optimum solution for the \kmo problem defined by the current point set $C$. Theorem~\ref{thm:local-search-apx-ratio} gives the following.

\begin{restatable}{claim}{claimofkmboundratio}
	\label{claim:ofkm-bound-ratio} At any moment immediately after the while loop (Loop~\ref{step:okm-while}), 
	we have $(1-\epsilon)\cost_p(S) \leq \big(3+\frac2\ell)\opt + \big(1+\frac1\ell\big)zp$.
\end{restatable}

\begin{restatable}{lemma}{lemmaboundp}
	\label{lemma:bound-p}
	At any moment, we have $p \leq \frac{2(3\ell+2)\opt}{\gamma(\ell+1)z}$. 
\end{restatable}

Combining Claim~\ref{claim:ofkm-bound-ratio} and Lemma~\ref{lemma:bound-p}, at the end of each time $t$, we have $(1-\epsilon)\cost_p(S) \leq \big(3+\frac2\ell\big)\opt + \big(1+\frac1\ell\big)zp \leq \big(3+\frac2\ell\big)\opt + \big(1+\frac1\ell\big)z \frac{2(3\ell+2)\opt}{\gamma(\ell+1)z} = \big(3+\frac2\ell\big)\opt + \big(3+\frac2\ell\big) \frac{2\opt}{\gamma} = \big(3+\frac2\ell\big)\big(1+\frac{2}{\gamma}\big)\opt$.  (This assumes $z \geq 1$, but the resulting inequality holds trivially when $z = 0$.) Defining the outliers to be the points $j$ with $d_p(j, S) = p$, our online algorithm achieves a bi-criteria approximation ratio of $\Big(\frac1{1-\epsilon}\big(1+\frac1\ell\big)(1+\gamma), \frac1{1-\epsilon}\big(3+\frac2\ell\big)\big(1+\frac{2}{\gamma}\big)\Big)$ since Step~\ref{step:check-outlier} guarantees that the solution $S$ has at most $\frac1{1-\epsilon}\big(1+\frac1\ell\big)(1+\gamma)z$ outliers.

\subsection{Analysis of Recourse} \label{subsec:Boundingrecourse}

We now proceed to the analysis of the total recourse of the online algorithm.   For simplicity, we define $\opt':=\min_{S' \subseteq F: |S'| = k } \cost_p(S')$ to be the cost of the optimum for the current $k$-median instance with metric $d_p$. Notice the difference between $\opt$ and $\opt'$: $\opt$ is for the original \kmo problem and  $\opt'$ is for the $k$-median with penalty problem (or $k$-median with metric $d_p$). So, $\opt'$ depends on both the current point set $C$ and the current $p$.  Like $\opt$, $\opt'$ can only increase during the course of the algorithm as $C$ only enlarges and $p$ only increases.

\begin{restatable}{claim}{claimboundpusingoptprime}
	\label{claim:bound-p-using-opt'}
	At any moment, we have $p \leq O(1) \cdot \opt'$.
\end{restatable}

We define a \emph{stage} of the online algorithm to be a period of the algorithm between two adjacent moments when we increase $p$ in Step~\ref{step:increase-p}. That is, a stage is an inclusion-wise maximal period of the algorithm in which the value of $p$ does not change.  From now on, we fix a stage and let $\bfp$ be the value of $p$ in the stage. So, $\bfp$ is fixed during the stage. Assume the stage starts at time $\tau$ and ends at time $\tau'$.  Notice that the stage contains the tail of time $\tau$, the head of time $\tau'$, and the entire time $\tau''$ for any $\tau'' \in [\tau + 1, \tau'-1]$.  An exceptional case is that $\tau = \tau'$, in which case the stage contains some period within time $\tau$. 


For every $t \in [\tau,\tau']$, let $\opt'_t$ be the optimum value for the $k$-median instance with $C = \{j_1, j_2, \cdots, j_t\}$ and metric $d_\bfp$.  So, for any $t \in [\tau, \tau']$, $\opt'_t$ is the value of $\opt'$ at any moment that is in the stage and after Step~\ref{step:include-j-t} at time $t$. For $t \in [\tau+1, \tau']$,  we define $\Delta_t$ to be the value of $\cost_{\bfp}(S)$ after Step~\ref{step:include-j-t} at time $t$, minus that before Step~\ref{step:include-j-t}. We view this as the increase of $\cost_\bfp(S)$ due to the arrival of $j_t$.   Let $\Delta_\tau$ be the value of $\cost_\bfp(S)$ at the beginning of the stage; that is, the moment immediately after $p$ is increased to $\bfp$.   

\begin{restatable}{lemma}{lemmakmbounddelta}
	\label{lemma:km-bound-delta}
	For every $T \in [\tau, \tau']$, we have 
		$\sum_{t = \tau}^{T}\Delta_t \leq O\left(k\log n\right) \opt'_{T}$.
\end{restatable}

We need the following technical lemma from \cite{GKLX20}.
\begin{restatable}{lemma}{helpersumba}
	\label{lemma:helper-sum-b/a}
	Let $b \in \R_{\geq 0}^H$ for some integer $H \geq 1$. Let $B_{H'} = \sum_{h=1}^{H'} b_h$ for every $H' = 0, 1, \cdots, H$. Let $0 < a_1 \leq a_2 \leq \cdots \leq a_H$ be a sequence of real numbers and $\alpha > 0$ such that $B_{H'} \leq \alpha \cdot a_{H'}$ for every $H' \in [H]$. Then we have 
		$\displaystyle \sum_{h = 1}^H \frac{b_h}{a_h} \leq \alpha \left(\ln \frac{a_H}{a_1} + 1\right)$.
\end{restatable}

We define $H = \tau' - \tau + 1$. For every $t \in [\tau', \tau]$, we define $b_{t-\tau'+1} = \Delta_t$ and $a_{t-\tau'+1} = \opt'_t$.  We define $B_{T-\tau'+1}$ for every $T = \tau-1, \tau, \cdots, \tau'$ to be $\sum_{t = \tau}^{T}b_{t-\tau'+1} = \sum_{t = \tau}^{T}\Delta_t$. By Lemma~\ref{lemma:km-bound-delta} we have $B_{H'} \leq \alpha\opt'_{H' + \tau-1} = \alpha \cdot a_{H'}$ for some $\alpha = O(k \log n)$ and every $H' \in [H]$.   In time $t$ within the stage, $\cost_{\bfp}(S)$ first increases by $\Delta_t$ in Step~\ref{step:include-j-t} (or becomes $\Delta_\tau$ at the beginning of the stage if $t = \tau$). Then for every median we swap inside the while loop~\ref{step:okm-while}, we decrease $\cost_{\bfp}(S)$ by at least $\frac{\epsilon\cost_{\bfp}(S)}{k} \geq \frac{\epsilon\cdot \opt'_{t}}{k}$, due to the use of the efficient swaps. Noticing that $\opt'_t$ is non-decreasing in $t$, using Lemma~\ref{lemma:helper-sum-b/a} we can bound the total recourse in the stage by
\begin{align*} \textstyle
	\sum_{t=\tau}^{\tau'} \frac{\Delta_t}{\epsilon\opt'_t/k}  = \frac k\epsilon\sum_{h=1}^{H} \frac{b_h}{a_h}  \leq
 \frac{k}{\epsilon}\alpha\left(\ln\frac{a_H}{a_1} + 1\right) = \frac{\alpha k}{\epsilon}\left(\ln\frac{\opt'_{\tau'}}{\opt'_{\tau}} + 1\right).
\end{align*}

Now it is time to consider all stages $[\tau, \tau']$ together.  The summation of $\ln \frac{\opt'_{\tau'}}{\opt_\tau}$ over all stages is the $\ln$ of the product of $\frac{\opt'_{\tau'}}{\opt'_\tau}$ over all stages.  For some time $t$ that crosses many different stages, $\opt'_t$ values depend on the $\bfp$ value of a stage. However as $\bfp$ increases, $\opt'_t$ can only increase. Therefore, the summation is at most $\ln$ of the ratio between the maximum possible $\opt'$ value and the minimum possible $\opt'$ value. So, this is at most $O(\log (n D))$.  There are most $\log_2 O(n D) = O(\log (nD))$ stages. Thus, the total recourse over the whole algorithm is at most $\frac{\alpha k}{\epsilon}\cdot O(\log (nD)) = O\left(\frac{k^2\log n \log (nD)}{\epsilon}\right)$. This finishes the proof of Lemma~\ref{lemma:okm-main}.

	\section{Experiments}
\label{sec:exp} 

In this section, we corroborate our theoretical findings by performing experiments on real world datasets.
Our goal is to empirically show that the local search algorithm is stable and does few reclusterings, while maintaining a good approximation factor.

\noindent{\bfseries Algorithm implementation:} We modified our algorithm slightly to make it faster: when a new data point comes, instead of conducting local search directly, we assign the point to its nearest center; then we check whether the current cost is at least $(1+\alpha)$ times the cost \emph{resulting from the last application of local search}, and if not we continue to the next data point without doing any local operations. It is easy to see that this will increase our approximation ratio by a $(1+\alpha)$ factor. Though this modification doesn't improve our worst-case recourse bound, it reduces the number of local operations needed when the incoming data are non-adversarial, which is often the case in practice. Throughout the experiment we set $\alpha=0.2$.

\noindent{\bfseries Data set and parameter setting:} Similar to \cite{LattanziV17}, we test the algorithm on three UCI data sets \cite{UCI-dataset}: (i) \textsc{Skin} with $245,057$ data points of dimension 4; (ii) \textsc{Covertype} with $581,012$ data points of dimension 54; In the experiment we'll only use the first $10$ features of \textsc{Covertype} because other features are categorical. (iii) \textsc{Letter} with $20,000$ data points of dimension 16. To keep the duration of experiments short, we restrict the experiments to the first 10K data points in each data set. We set the algorithm parameters $\epsilon=0.05$ and $\gamma=1$; these were chosen to minimize the number of discarded outliers. We set the available center locations $F=C$, so when a new data point comes, it will be added to both $F$ and $C$. Throughout the experiment, we set the number of outliers to be $z=200$, and tried three different values of $k\in\{10, 50, 100\}$. We observe that in all the runs, our algorithm removes at most $840$ outliers, hence achieving an approximation factor of $4.2$ on the number of discarded outliers.


\noindent{\bfseries Results:} We first show the how the recourse grows overtime in Figure~\ref{fig:recourse}. One can observe that the recourse dependence on $k$ is roughly $O(k \log n)$ instead of the $O(k^2 \log n \log (nD))$ worst-case bound predicted by our theoretical result. 
We also observe that the growth rate of recourse is lower for \textsc{Covertype} and \textsc{Letter} data sets compared to \textsc{Skin}. 
This is because of the data ordering in \textsc{Skin};
if we randomly shuffle the \textsc{Skin} data set and re-run the algorithm then we get a graph similar to the other two data sets.


\begin{figure}[!htb]
  \centering
  \subfloat[\textsc{Skin}]{
    \includegraphics[width=0.32\textwidth]{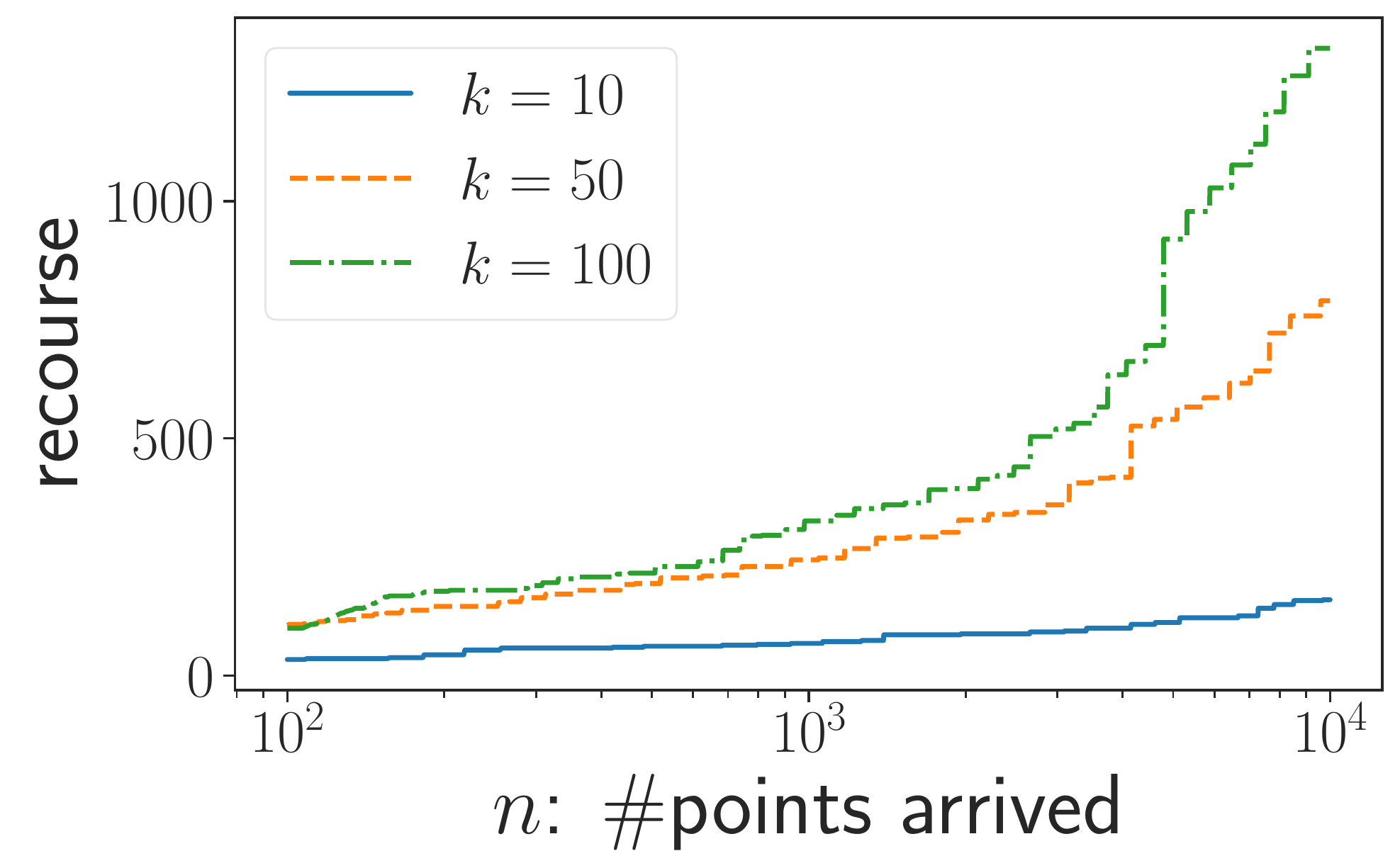}
  }
  ~
  \subfloat[\textsc{Covertype}]{
    \includegraphics[width=0.32\textwidth]{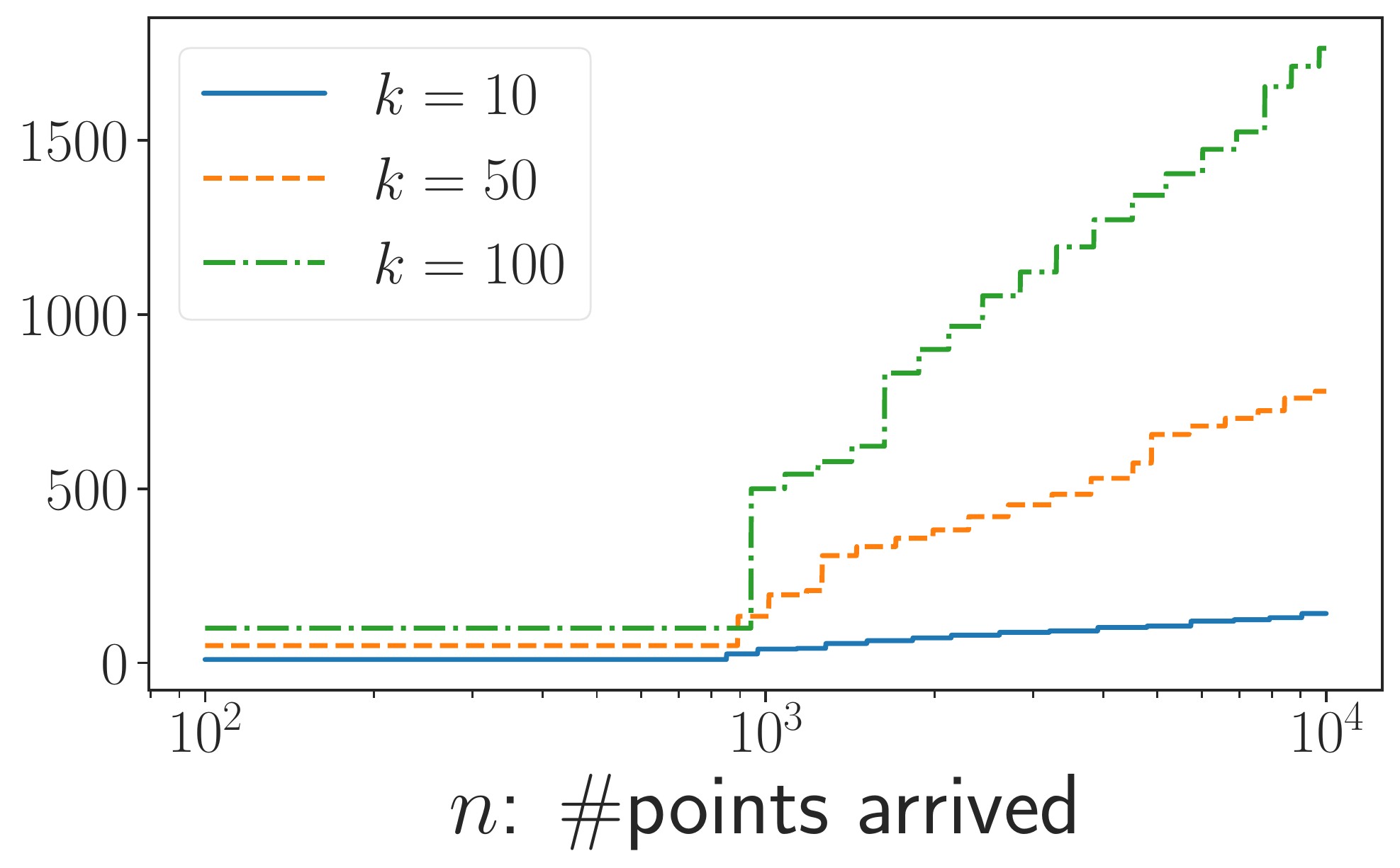}
  }
  ~
  \subfloat[\textsc{Letter}]{
    \includegraphics[width=0.32\textwidth]{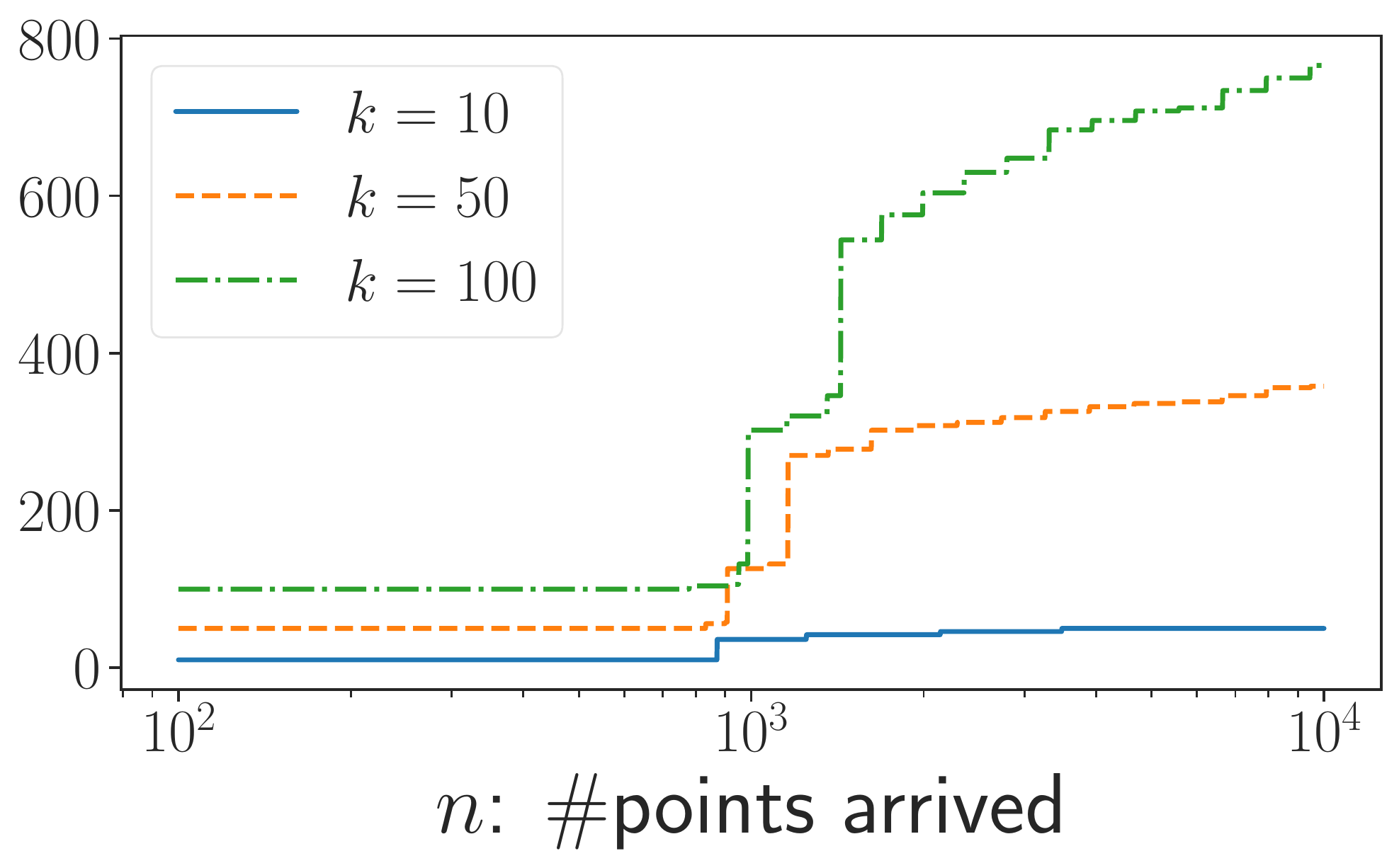}
  }
  \caption{Recourse over time. The $x$-axis is plotted in the log-scale}
  \label{fig:recourse}
\end{figure}
\begin{figure}[!htb]
  \centering
  \subfloat[\textsc{Skin}]{
    \includegraphics[width=0.32\textwidth]{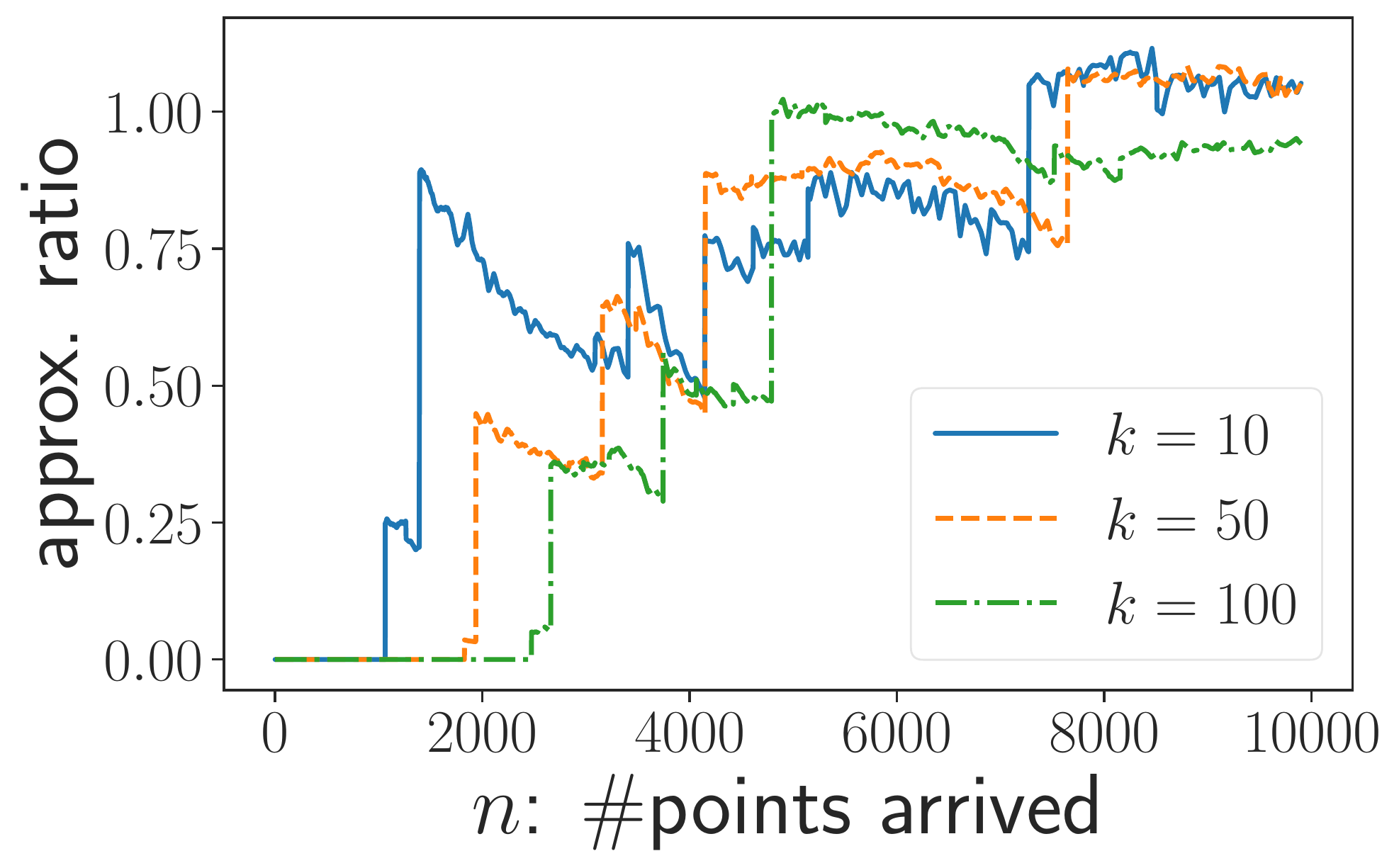}
  }
  \subfloat[\textsc{Covertype}]{
    \includegraphics[width=0.32\textwidth]{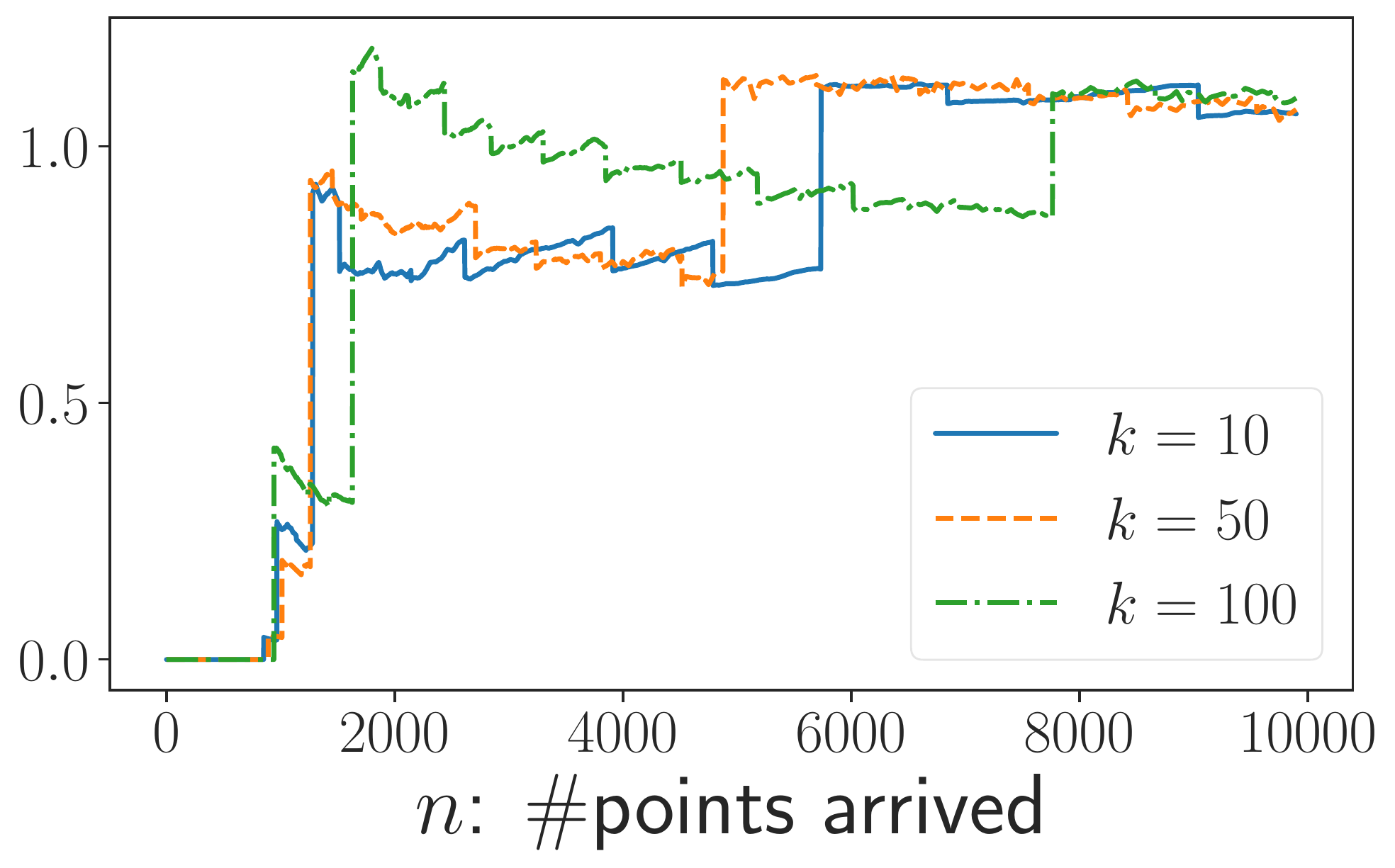}
  }
  \subfloat[\textsc{Letter}]{
    \includegraphics[width=0.32\textwidth]{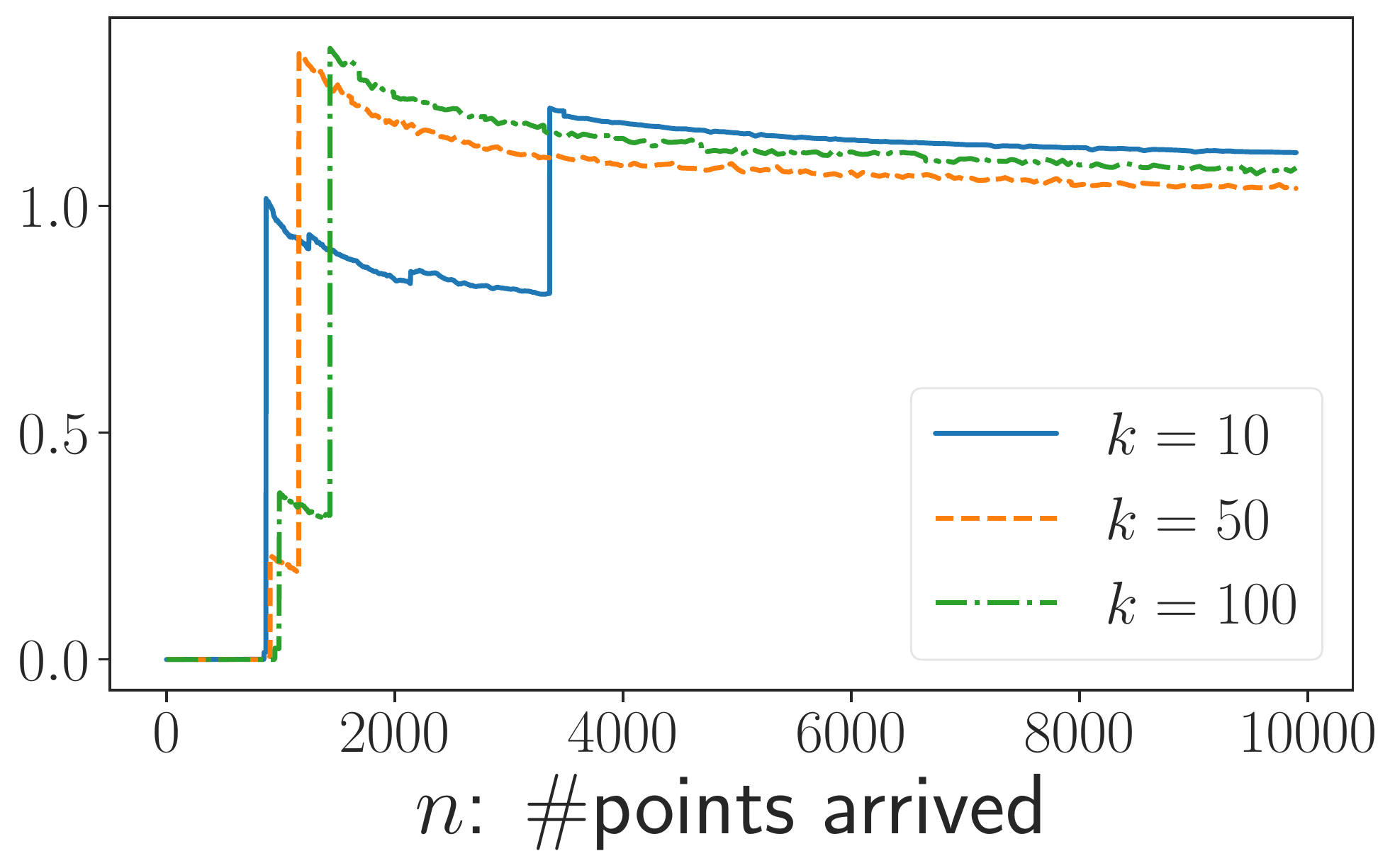}
  }
  \caption{Estimated approximation ratio over time.}
  \label{fig:cost}
\end{figure}
Now we turn to the quality of clustering maintained by our algorithm. Since the optimal solution is hard to compute, we use the clustering produced by offline $k$-means$--$ algorithm of \cite{ChawlaGionis13} as an coarse estimation of OPT. Specifically, for every 50 newly-arrived data points, we compute 5 offline $k$-means$--$ solutions (with different initializations) for all already arrived data points, and choose the best one as the estimation for OPT at this time point. Then we linearly interpolate between these estimations to get an OPT curve for every time point.
Figure~\ref{fig:cost} shows the estimated approximation ratio over time. {\em We see that the ratio is bounded by $1.5$ most of the times.} One might notice that the ratio sometimes even falls below 1. This is because of two reasons: 1) we only have an estimate of the real OPT; 2) the bi-criteria approximation means our algorithm might remove more than $z=200$ outliers, while the OPT is calculated by removing at most $z$ outliers.

Lastly, we also ran our experiments by allowing $z$ to increase over time and noticed similar behavior. Due to space constraints, we give those results in the supplementary material.

	\bibliographystyle{plain}
	\bibliography{main}

	\appendix
	\newpage
	\setcounter{page}{1}
	\pagenumbering{roman}
	\setcounter{figure}{0}
	\renewcommand{\thefigure}{\roman{figure}}
	\setcounter{theorem}{0}
	\renewcommand{\thelemma}{\Roman{theorem}}
	
	
\section{Missing Proofs from Section \ref{sec:okm}}
In this section we prove Theorem~\ref{thm:local-search-apx-ratio}.
\adjuflapprox*
\begin{proof}
 By making copies of medians, we assume $S$ and $S^*$ are disjoint. 
 For every $j \in C$, define $\sigma(j)$ and $\sigma^*(j)$ to be the closest median of $j$ in $S$ and $S^*$ respectively.  Let $O^* = \big\{j: d_p(j, S^*) \geq \frac{\ell+1}{3\ell+2}p\big\}$; these are the points $j$ with $\min\Big\{\big(3+\frac2\ell\big)d_p(j, S^*), \big(1+\frac1\ell\big)p\Big\} = \big(1+\frac1\ell\big)p$. For every $i^* \in S^*$, define $\phi(i^*)$ to be the nearest median of $i^*$ in $S$, according to the metric $d_p$, breaking ties arbitrarily. 
  We partition $S$ into three parts as follows:
\begin{itemize}[leftmargin=*,topsep=0pt,itemsep=0pt]
	\item $S_0 := \{i \in S: \phi^{-1}(i) = \emptyset\}$.
    \item $S_{1}:=\{i\in S: 1 \leq |\phi^{-1}(i)|\leq \ell\}$. 
  \item $S_{+}:=\{i\in S: |\phi^{-1}(i)| > \ell\}$. 
\end{itemize}

Let $S_{1}^*:=\phi^{-1}(S_{1})$ (which is defined as $\union_{i \in S_1}\phi^{-1}(i)$) and $S_{+}^*:=\phi^{-1}(S_{+})$; thus $(S_{1}^*, S_{+}^*)$ is a partition of $S^*$. Moreover,$|S_{1}| \leq |S^*_{1}|$ and $|S_{+}|\leq |S_{+}^*|/(\ell+1)$. This implies 
\begin{align}
	|S_0| &= k-|S_{1}| - |S_{+}|  \geq (|S^*_1| - |S_1|) +  (k - |S^*_{1}|)  - |S_{+}^*|/(\ell+1) \nonumber\\
	&= (|S^*_1| - |S_1|) + |S^*_+| - |S_{+}^*|/(\ell+1) = |S^*_1| - |S_1| + \frac{\ell}{\ell+1}|S^*_+|. \label{inequ:S0-big}
\end{align} 

We define a random mapping $\beta: S^* \to S_0 \cup S_1$ in the following way. See Figure~\ref{fig:swap} for the illustration of the procedure. We first define $\beta$ over $S^*_1$.  For every $i \in S_1$, we take an arbitrary $i^* \in \phi^{-1}(i)$ and define $\beta(i^*) = i$; for all other facilities ${i^*}'$ in $\phi^{-1}(i)$, we define $\beta({i^*}')$ to be an arbitrary median in $S_0$.   So, $|S_1|$ medians in $S^*_1$ are mapped to $S_1$ by $\beta$ and the remaining $|S^*_1| - |S_1|$ facilities in $S^*_1$ are mapped to $S_0$.  By \eqref{inequ:S0-big}, we can make $\beta$ restricted to $S^*_1$ an injective function. Moreover, at least $\frac{\ell}{\ell+1}|S^*_+|$ facilities in $S_0$ do not have preimages so far; call the facilities free facilities. Then, we map $S^*_+$ to these free facilities in a random way so that each free facility is mapped to at most twice and in expectation, each free facility in expectation has at most $\big(1+\frac1\ell\big)$ pre-images in the function $\beta$.

\begin{figure}[h]
	\centering
	\includegraphics[width=0.5\textwidth]{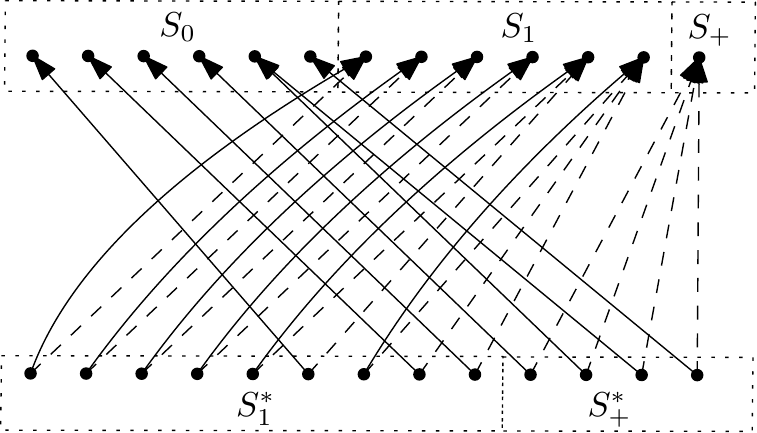} 
	\caption{The definition of the function $\beta$. The vertices at the top are $S$, the vertices at the bottom $S^*$, $\ell = 3$, and the dashed lines give the definition of $\phi$.   Then $S_0, S_1, S_+, S^*_1, S^*_+$ are depicted in the figure, and a possible function $\beta$ is given by the solid lines and curves.} \label{fig:swap}
\end{figure}

With the random $\beta$ defined,  we describe a set of \emph{test swaps} that will be used in our analysis. For every $i \in S_1$, we have a test swap $(\phi^{-1}(i), \beta(\phi^{-1}(i)))$.  For every $i^* \in S^*_{+}$, we have a test swap $(\{i^*\}, \{\beta(i^*)\})$. It is easy to see that each test swap $(A^*, A)$ has $A^* \subseteq F^*, A \subseteq F$ and $|A^*| = |A| \leq \ell$. Moreover, we have the following properties:
\begin{enumerate}[label=(P\arabic*)]
	\item Every median in $i^* \in S^*$ is swapped in exactly once in all test swaps.
	\item In expectation over all possible $\beta$'s, every median in $i \in S$ is swapped out at most $1 + \frac{1}{\ell}$ times in the test swaps. 
	\item For any test swap $(A^*, A)$, we have $\phi^{-1}(A) \subseteq A^*$.
\end{enumerate}
(P1) and (P2) follow from the construction of $\beta$. To see (P3), consider the two types of test swaps. If the test swap is $(\{i^*\}, \{\beta(i^*)\})$ for some $i^* \in S^*_+$, then $\beta(i^*) \in S_0$ and thus $\phi^{-1}(\beta(i^*)) = \emptyset$. If the test swap is $(\phi^{-1}(i), \beta(\phi^{-1}(i)))$ for some $i \in S_1$, then $\beta(\phi^{-1}(i))$ contains $i$ and all the other elements in the set are in $S_0$. Thus $\phi^{-1}(\beta(\phi^{-1}(i))) = \phi^{-1}(i)$.


 Focus on a fixed test swap $(A^*, A)$.  After opening $A^*$ and closing $A$, we can reconnect a subset of points in $\sigma^{-1}(A) \cup \sigma^{*-1}(A^*)$. We guarantee that all points in $\sigma^{-1}(j)$ will be reconnected. See Figure~\ref{fig:reconnection} for how we reconnect the points. 
 \begin{figure}
 	\centering
 	\includegraphics[width=0.8\textwidth]{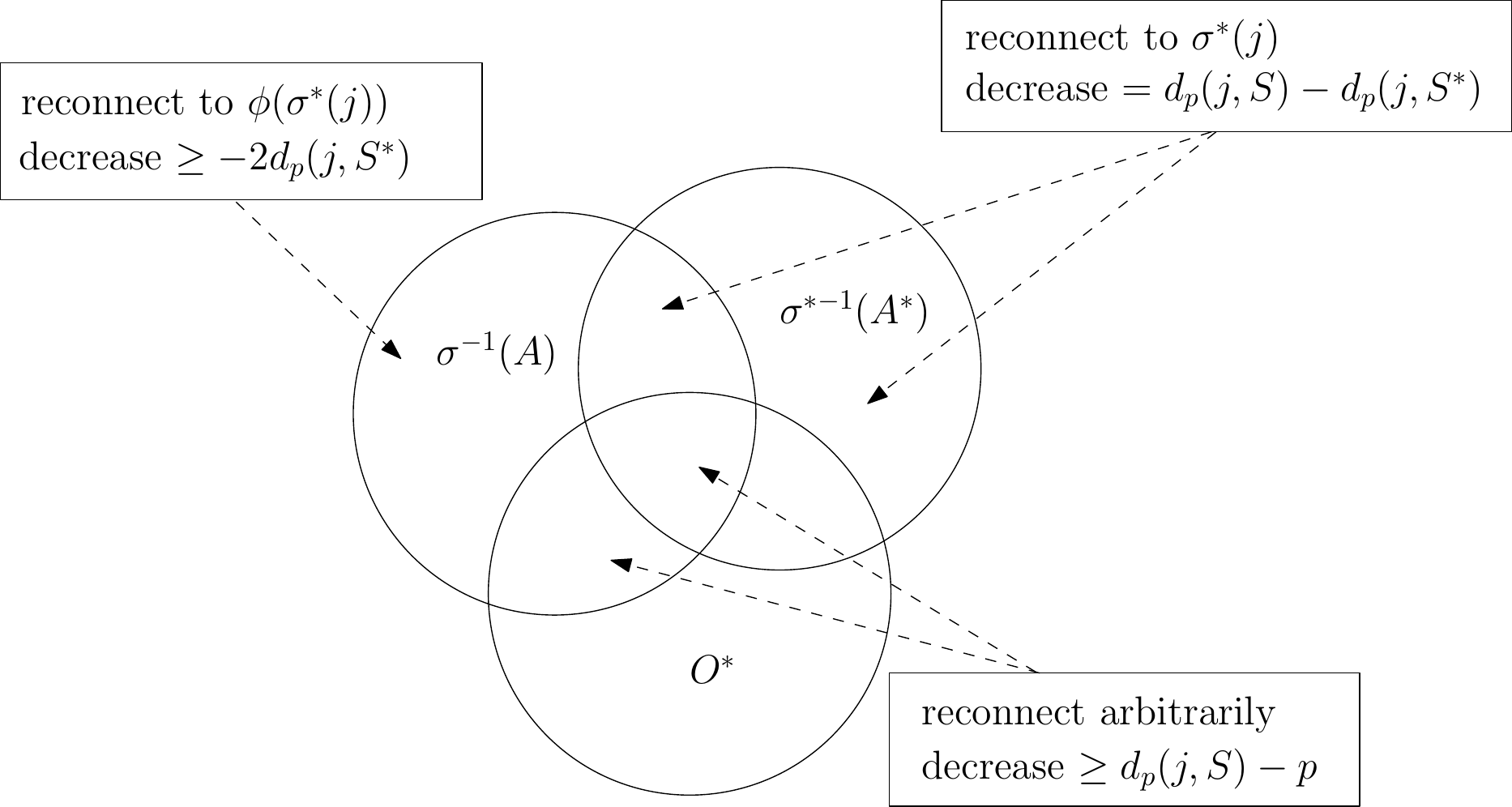}
 	\caption{How to reconnect points and the lower bound for the decrement in the connection cost for each point $j$, using the Venn diagram for the three sets $\sigma^{-1}(A), \sigma^{*-1}(A^*)$ and $O^*$.} \label{fig:reconnection}
 \end{figure}
 \begin{itemize}[leftmargin=*]
 	\item For a point $j \in \sigma^{*-1}(A^*) \setminus O^*$, we reconnect $j$ from $\sigma(j)$ to $\sigma^*(j) \in A^*$.  The decrease in the connection cost of $j$ is $d_p(j, \sigma(j)) - d_p(j, \sigma^*(j)) = d_p(j, S) - d_p(j, S^*)$.
 	\item For a point $j \in \sigma^{-1}(A) \setminus \sigma^{*-1}(A^*) \setminus O^*$, we reconnect $j$ to $\phi(\sigma^*(j))$.  Notice that $\sigma^*(j) \notin A^*$. By (P3), we have $\phi(\sigma^*(j)) \notin A$. Thus the connection is valid.  By triangle inequalities and definition of $\phi$, for every $j \in \sigma^{-1}(A) \setminus \sigma^{*-1}(A^*)\setminus O^*$, we have 
 	\begin{align*}
 		d_p(j, \phi(\sigma^*(j)))&\leq d_p(j, \sigma^*(j)) + d_p(\sigma^*(j), \phi(\sigma^*(j))) \leq d_p(j, \sigma^*(j)) + d_p(\sigma^*(j), \sigma(j))\\
 		& \leq d_p(j, \sigma^*(j)) + d_p(j, \sigma^*(j)) + d_p(j, \sigma(j)) = 2d_p(j, \sigma^*(j))+d_p(\sigma(j),j).
 	\end{align*}
 	So the decrease in the connection cost of $j$ is $d_p(j, \sigma(j)) - d_p(j, \phi(\sigma^*(j))) \geq -2d_p(j, \sigma^*(j)) = -2d_p(j, S^*)$.
 	\item For a point $j \in \sigma^{-1}(A) \cap O^*$, we reconnect $j$ arbitrarily, and the decrease in the connection cost of $j$ is at least $d_p(j, \sigma(j)) - p = d(j, S) - p$ as $p$ is the diameter of the metric $d_p$. 
 \end{itemize}
 
As the test swap operation is not $\rho$-efficient, we have 
%
\begin{align}
	  \quad\sum_{j \in \sigma^{*-1}(A^*) \setminus O^*}(d_p(j, S)-d_p(j, S^*)) - 2\sum_{j \in \sigma^{-1}(A) \setminus O^*}d_p(j, S^*) \nonumber&\\
	  + \sum_{j \in \sigma^{-1}(A)\cap O^*}(d_p(j, S) - p)&\leq |A|\rho. \label{eq:test-swap-witness}
\end{align}
Above, we used that  that $\sigma^{-1}(A)\setminus \sigma^{*-1}(A^*) \setminus O^*\subseteq \sigma^{-1}(A) \setminus O^*$. 


We now add up \eqref{eq:test-swap-witness} over all test swap operations.  We consider the expectation of the left side of the summation, over all random choices of $\beta$:
\begin{itemize}[leftmargin=*]
	\item The sum of the first term on the left side of \eqref{eq:test-swap-witness} is always exactly $\sum_{j \in C \setminus O^*} \big(d_p(j, S) - d_p(j, S^*)\big)$, due to (P1).
	\item Consider the expectation of the sum of the second term on the left side of \eqref{eq:test-swap-witness}. Since each $i \in S$ is swapped out in at most $1+\frac1\ell$ times in expectation by (P2), the expectation of the sum of the second term is at least $-\left(2 + \frac2\ell\right) \sum_{j \in C \setminus O^*} d_p(j, S^*)$.
	\item Consider the expectation of the sum of the third term on the left side of \eqref{eq:test-swap-witness}. Using that $d_p$ has diameter at most $p$, and (P2), the expectation is at least $\left(1+\frac1\ell\right)\sum_{j \in O^*} (d_p(j, S) -p) \geq \sum_{j \in O^*}d_p(j, S) - \left(1+\frac1\ell\right)|O^*|p$. We changed the coefficient before a non-negative term from $\left(1+\frac1\ell\right)$ to $1$ in the inequality; this is sufficient. 
\end{itemize}
Overall, the expectation of the sum of the left side of \eqref{eq:test-swap-witness} over all test swap operations is at least 
\begin{align*}
	&\sum_{j \in C \setminus O^*}\big(d_p(j, S) - d_p(j, S^*)\big) - \left(2+\frac2\ell\right) \sum_{j \in C \setminus O^*} d_p(j, S^*)+ \sum_{j \in O^*}d_p(j, S)  - \left(1+\frac1\ell\right)|O^*|p\\
	&= \sum_{j \in C}d_p(j, S)- \left(3+\frac2\ell\right) \sum _{j \in C\setminus O^*}d_p(j, S^*) - |O^*|\cdot \left(1+\frac1\ell\right)p \\
	&= \sum_{j \in C}d_p(j, S) - \sum_{j \in C} \min\left\{\left(3+\frac2\ell\right)d_p(j, S^*), \left(1+\frac1\ell\right)p\right\},
\end{align*}
where the last equality used the definition of $O^*$. 

The summation of the right side of \eqref{eq:test-swap-witness} over all test swaps is always exactly $k\rho $.  Therefore, we have 
\begin{align*}
	 \sum_{j \in C}d_p(j, S) - \sum_{j \in C} \min\left\{\big(3+\frac2\ell\big)d_p(j, S^*), \big(1+\frac1\ell\big)p\right\} \leq  k \rho .
\end{align*}
Rearranging the terms and replacing $ \sum_{j \in C}d_p(j, S)$ with $\cost_p(S)$ finish the proof of the theorem. 
\end{proof}

\section{Missing Proofs from Section \ref{subsec:BoundingApprox}} 
\claimofkmboundratio*
\begin{proof} 
    After the while loop, no $\frac{\epsilon \cdot \cost_p(S)}{k}$-efficient swaps can be performed. Applying Theorem~\ref{thm:local-search-apx-ratio} with $S^*$ being the optimum solution for the \kmo instance at the moment,  we have $\cost_{p}(S) \leq 0.1+\sum_{j \in C}\min\big\{\big(3+\frac2\ell\big)d_p(j, S^*), \big(1+\frac1\ell\big)p\big\} +k \cdot  \frac{\epsilon \cdot \cost_p(S)}{k} \leq \big(3+\frac2\ell\big)\opt +\big(1+\frac1\ell\big)zp + \epsilon\cdot \cost_p(S)$. Moving $\epsilon\cdot \cost_p(S)$ to the left side gives the claim. 
\end{proof}

\lemmaboundp*
\begin{proof}
	The statement holds at the beginning since $\opt = 0.1$ and $p \leq 0.1$.
	As $\opt$ can only increase during the algorithm, it suffices to prove the inequality at any moment after we run Step~\ref{step:increase-p}; this is the only step in which we increase $p$. 	We assume $z \geq 1$ since if $z = 0$ the lemma is trivial. 
	
	Focus on any moment before we run Step~\ref{step:increase-p}. We define $p^* > 0$ to be the real number such that $\big(1+\frac1\ell\big)(1+\gamma)zp^* = \big(3+\frac2\ell\big)\opt +\big(1+\frac1\ell\big)zp^*$. Then, if $p > p^*$, then the condition in Step~\ref{step:check-outlier} does not hold: Otherwise, we have $(1-\epsilon)\cost_p(S) > (1-\epsilon)\cdot \frac{1}{1-\epsilon}\big(1+\frac1\ell\big)(1+\gamma)zp \geq \big(3+\frac2\ell\big)\opt + \big(1+\frac1\ell\big)zp$, contradicting Claim~\ref{claim:ofkm-bound-ratio}. Since we assumed we are going to run Step~\ref{step:increase-p},  we have $p \leq p^*$.  So, after Step~\ref{step:increase-p}, we have $p \leq 2p^* = 2\cdot\frac{(3+2/\ell)\opt}{\gamma(1+1/\ell)z} = \frac{2(3\ell+2)\opt}{\gamma(\ell+1)z}$.
	%
\end{proof}

\section{Missing Proofs from Section \ref{subsec:Boundingrecourse}} 
\claimboundpusingoptprime*
\begin{proof}
	Again it suffices to show the inequality at any moment after we run Step~\ref{step:increase-p}. Suppose we just completed the while loop. Applying Theorem~\ref{thm:local-search-apx-ratio} with $S^*$ being the optimum solution for the current $k$-median instance with metric $d_p$,  we have $\cost_p(S) \leq \frac{1}{1-\epsilon}\big(3+\frac{2}{\ell}\big)\opt'$. If at the moment we have $p > \frac{1}{1-\epsilon}\big(3+\frac{2}{\ell}\big)\opt'$, then the condition for Step~\ref{step:check-outlier} will not be satisfied, even if $z = 0$. So, before we run Step~\ref{step:increase-p}, we must have $p \leq \frac{1}{1-\epsilon}\big(3+\frac{2}{\ell}\big)\opt'$. After the step, we have $p \leq \frac2{1-\epsilon}\big(3+\frac{2}{\ell}\big)\opt' = O(1)\cdot \opt'$. 
\end{proof}

\lemmakmbounddelta*
\begin{proof}
	We can show that $\Delta_\tau \leq O(1) \opt'_{\tau} \leq O(1)\opt'_{T}$ by applying Theorem~\ref{thm:local-search-apx-ratio} with $S^*$ being the optimum solution that defines $\opt'_{\tau}$. Thus it suffices to bound $\sum_{t = \tau+1}^{T}\Delta_t$.  
	
	Let $S^*$ be the optimum solution for the $k$-median instance with point set $\{j_1, j_2, \cdots, j_{T}\}$ and metric $d_\bfp$. We are only interested in points $j_{\tau+1}, j_{\tau+2}, \cdots, j_{T}$ in the analysis.   Fix any $i^* \in S^*$.  Let $\{j_{t_1}, j_{t_2}, \cdots, j_{t_s}\}$ be the set of points in $\{j_{\tau+1}, j_{\tau+2}, \cdots, j_{T}\}$ connected to $i^*$ in the solution $S^*$, where $\tau < t_1 < t_2 < \cdots < t_s \leq T \leq \tau$.  For notation convenience, we let $j'_r = j_{t_r}$ and $\Delta'_r = \Delta_{t_r}$ for every $r \in [s]$.  We now bound $\sum_{r = 1}^s\Delta'_{r}$.  We assume $s \geq 1$ since otherwise the quantity is $0$.
	
	We can bound $\Delta'_1$ by $\bfp$, and by Claim~\ref{claim:bound-p-using-opt'}, we have $\Delta'_1 \leq \bfp \leq O(1)\cdot \opt'_\tau \leq O(1) \cdot \opt'_{T}$.   Then we will bound $\Delta'_r$ for any integer $r \in [2, s]$.  Using Theorem~\ref{thm:local-search-apx-ratio}, we can show that at the beginning of time $t_r$ (or equivalently, at the end of time $t_r-1$), we have $\cost_\bfp(S) \leq O(1) \cdot \opt'_{t_r-1} \leq O(1) \cdot \opt'_{T}$. For the $S$, we have
\begin{align*} \textstyle
	\sum_{u = 1}^{r-1}\big(d_\bfp(j'_{u}, S) + d_\bfp(j'_{u}, i^*)\big) \leq O(1) \opt'_{T}.
\end{align*} 
The inequality holds since the summation for each of the two terms is at most $O(1) \opt'_{T}$.  So, there is at least one point $j'_u$ such that $d_\bfp(j'_u, S) + d_\bfp(j'_u, i^*) \leq O(1)\cdot \frac{\opt'_{T}}{r-1}$, implying $d_\bfp(i^*, S) \leq O(1)\cdot \frac{\opt'_{T}}{r-1}$.  Therefore, we have $\Delta'_r \leq d_\bfp(i^*, j'_r) + d_\bfp(i^*, S) \leq d_\bfp(i^*, j'_r) + O(1)\cdot \frac{\opt'_{T}}{r-1}$. Then
	\begin{align*} \textstyle
		\sum_{r = 1}^s\Delta'_r &\textstyle \leq O(1)\cdot \opt'_{T} + \sum_{r = 2}^s \left(d_\bfp(i^*,j'_r) + O(1)\cdot\frac{\opt'_{T}}{r-1}\right) \\
		 &\textstyle\leq \sum_{r=1}^s d_\bfp(i^*, j'_r) + O(\log s) \opt'_{T}=O(\log T) \opt'_{T} = O(\log n) \opt'_{T}.
	\end{align*}
	Considering all the $k$ medians $i^* \in S^*$ together, we have $\sum_{t = \tau+1}^{T} \Delta_t \leq O(k \log n)\opt'_{T}$.
\end{proof}

\helpersumba*
\begin{proof} Define $a_{H+1} = +\infty$.
	\begin{align*}
		\sum_{h = 1}^H \frac{b_h}{a_h} &= \sum_{h = 1}^H \frac{B_h - B_{h-1}}{ a_{h}}=\sum_{h = 1}^{H} B_h \left(\frac{1}{a_h} - \frac{1}{a_{h+1}}\right) = \sum_{h = 1}^{H}\frac{B_h}{a_h} \left(1 - \frac{a_h}{a_{h+1}}\right) \leq \alpha \sum_{h = 1}^{H}\left(1 - \frac{a_{h}}{a_{h+1}}\right)\\
		&=\alpha H- \alpha\sum_{h = 1}^{H-1}\frac{a_h}{a_{h+1}} \leq \alpha H- \alpha(H-1)\Big(\frac{a_1}{a_H}\Big)^{1/(H-1)} \\
		&=  \alpha(H-1)\left(1-e^{-\ln\frac{a_H}{a_1}/(H-1)}\right) + \alpha \leq \alpha(H-1)\ln\frac{a_H}{a_1}/(H-1) + \alpha = \alpha\left(\ln \frac{a_H}{a_1}+1\right).
	\end{align*}
	The inequality in the second line used the following fact: if the product of $H-1$ positive numbers is $\frac{a_1}{a_H}$, then their sum is minimized when they are equal.  The inequality in the third line used that $1-e^{-x} \leq x$ for every $x$.
\end{proof}

\section{Handling the $F = C$ Setting} \label{subsec:F=C}
When $F = C$,  a small issue with the analysis is that $\opt$ and $\opt'$ may decrease as the algorithm proceeds. However, it can only decrease by at most a factor of $2$  from a moment to any later moment. This holds due to the following fact:  If we have a star $(i, C')$ and any metric $d'$, we have $\min_{j^* \in C'} \sum_{j \in C'}d(j^*, j) \leq 2 \sum_{j \in C'} d'(i, j)$. That is, including additional medians in $F$ on top of $F=C$ can only save a factor of $2$.  

To address the issue, we define $\opt$ to be the optimum value of the current $\kmo$ instance. We define $\overline{\opt}$ at any moment of the algorithm to be the maximum $\opt$ we see until the moment.  Then at any moment of the algorithm, we have $\opt \leq \overline{\opt} \leq 2\opt$.   Moreover $\overline{\opt}$ can only increase as the algorithm proceeds.  Claim~\ref{claim:ofkm-bound-ratio} still holds, and Lemma~\ref{lemma:bound-p} holds with $\opt$ replaced by $\overline{\opt}$ or $2\opt$.  Then eventually we shall get a bifactor of $\left(\frac{1}{1-\epsilon}\big(1+\frac1\ell\big)(1+\gamma),\frac{1}{1-\epsilon}(3 + \frac2\ell)\big(1+\frac{4}{\gamma}\big)\right)$.

We can use the same trick to handle $\opt'$ in the analysis of the recourse. In this case, the factor of $2$ will be hidden in the $O(\cdot)$ notation and thus the recourse bound is not affected.  More precisely, we define $\overline{\opt}'$ to be the maximum $\opt'$ we see until the moment.  Then, we always have $\opt' \leq \overline{\opt}' \leq 2\opt'$, and $\overline{\opt}'$ can only increase.  Claim~\ref{claim:bound-p-using-opt'} still holds.   Then we fix a stage whose $p$ value is $\bfp$ and assume the stage starts in time $\tau$ and ends in time $\tau'$.  For every $t \in [\tau, \tau']$, define $\overline{\opt}'_t$ to be the value of $\overline{\opt}'$ at any moment that is in the stage and after Step~\ref{step:include-j-t} at the time $t$. Then Lemma~\ref{lemma:km-bound-delta} still holds and in the end we can bound the recourse by  $O\left(\frac{k^2\log n \log (nD)}{\epsilon}\right)$. Thus we proved Lemma~\ref{lemma:okm-main}.

\section{Additional Experiment Results for Incremental $z$ setting}
Here we include experiment results for the incremental $z$ setting where the number of outliers $z$ changes with time. We let $z$ grow uniformly as follows: we still focus on the first 10K data points, and for each time $t\in[1, 10000]$, we set the number of allowed outliers $z_t=\frac{t}{10000}\times 200$. So as more data points come, we allow to remove more outliers. All other parameters are the same as in section~\ref{sec:exp}: $\epsilon=0.05, \gamma=1, k\in\{10, 50, 100\}$, and available center locations $F=C$.

\begin{figure}[!htb]
  \centering
  \subfloat[\textsc{Skin}]{
    \includegraphics[width=0.32\textwidth]{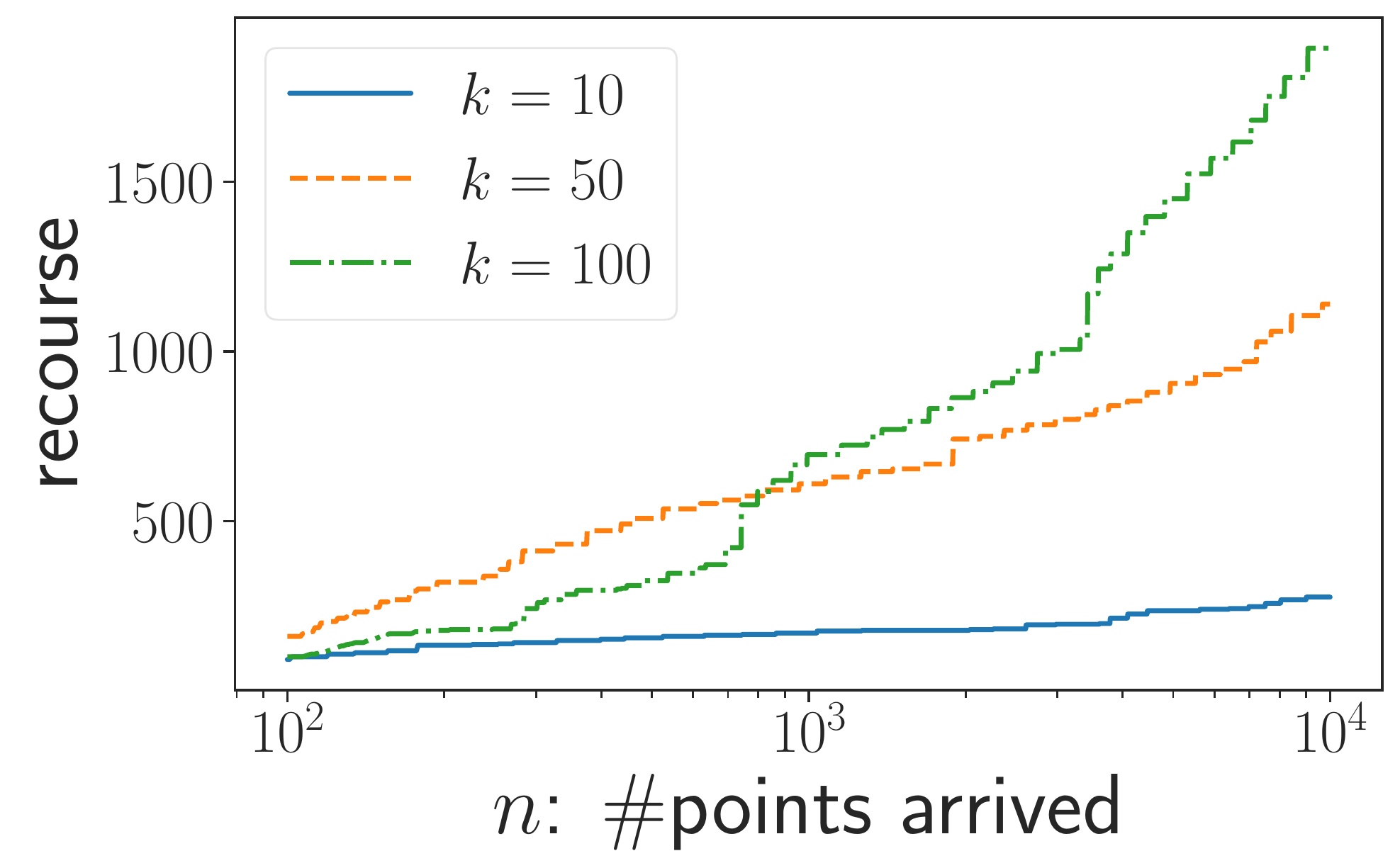}
  }
  ~
  \subfloat[\textsc{Covertype}]{
    \includegraphics[width=0.32\textwidth]{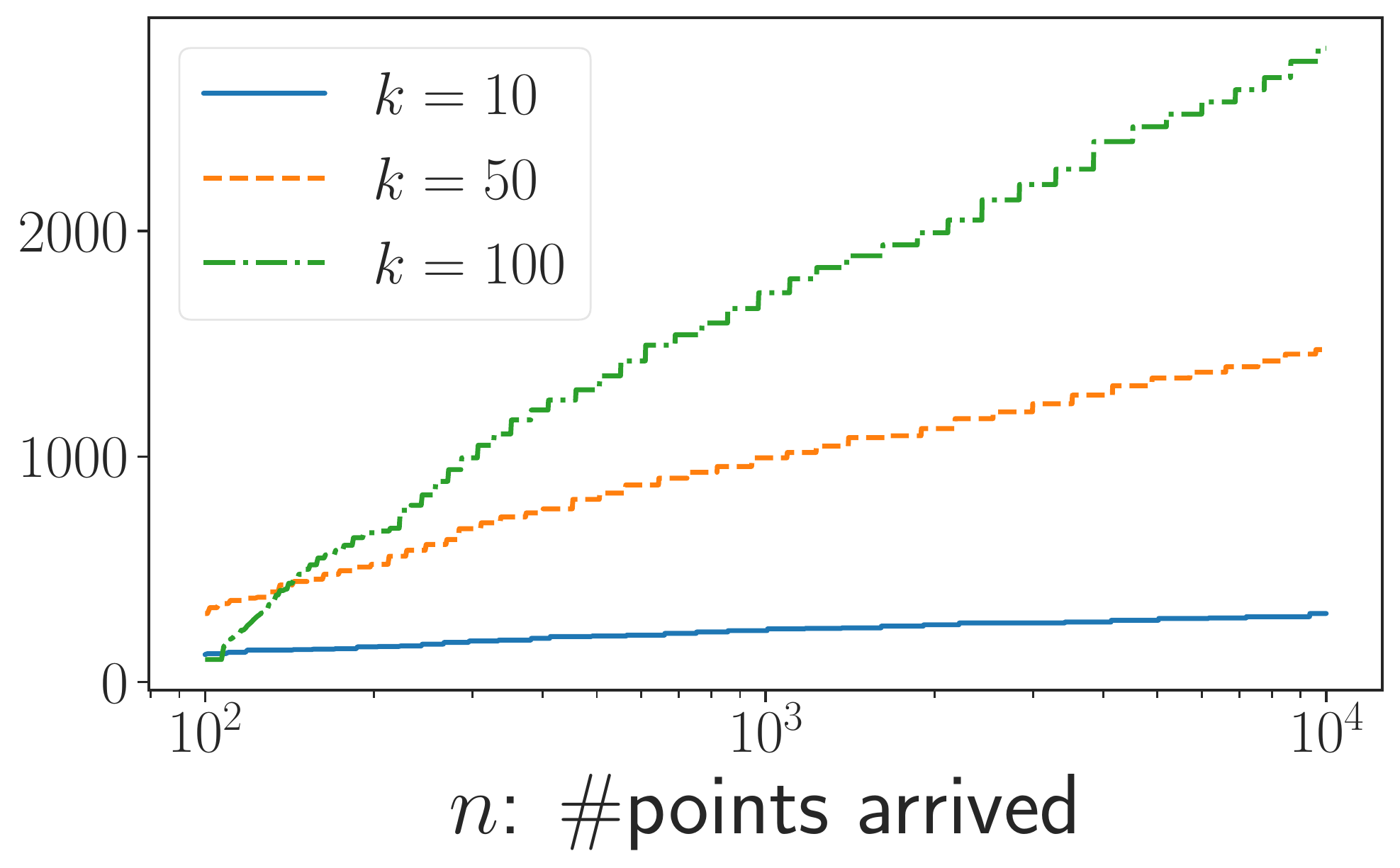}
  }
  ~
  \subfloat[\textsc{Letter}]{
    \includegraphics[width=0.32\textwidth]{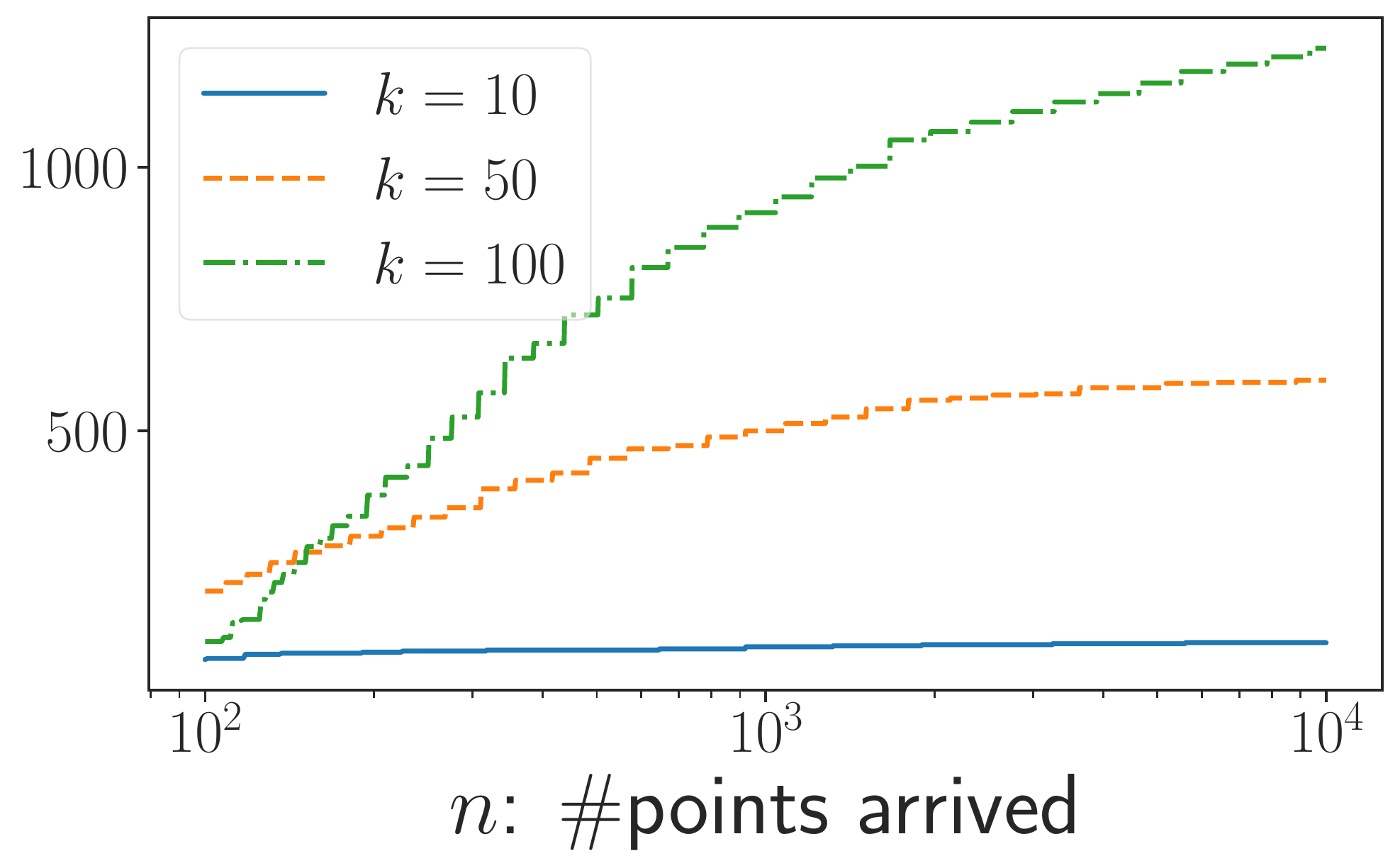}
  }
  \caption{Recourse over time. The $x$-axis is plotted in the log-scale}
  \label{fig:recourse-changing-z}
\end{figure}

Figure~\ref{fig:recourse-changing-z} shows how the total recourse grows with time. One can see that it's largely the same as that in Figure~\ref{fig:recourse}, exhibiting an $O(k)$ dependence on $k$ and $O(\log n)$ dependence on $n$. The major difference is that the recourse starts growing in very early time stages, while in Figure~\ref{fig:recourse} there's a longer warm-up phase. This is because in the setting of Figure~\ref{fig:recourse} the algorithm is allowed to remove roughly $4z=800$ outliers from the beginning, which means it can simply ignore the first few hundred arrived data points and conduct no local operations, i.e., no recourse. Figure~\ref{fig:cost-changing-z} shows the clustering quality on the three data sets. One can see that our algorithm still achieves very good approximation ratio (nearly 1) on all three data sets.

\begin{figure}[!htb]
  \centering
  \subfloat[\textsc{Skin}]{
    \includegraphics[width=0.32\textwidth]{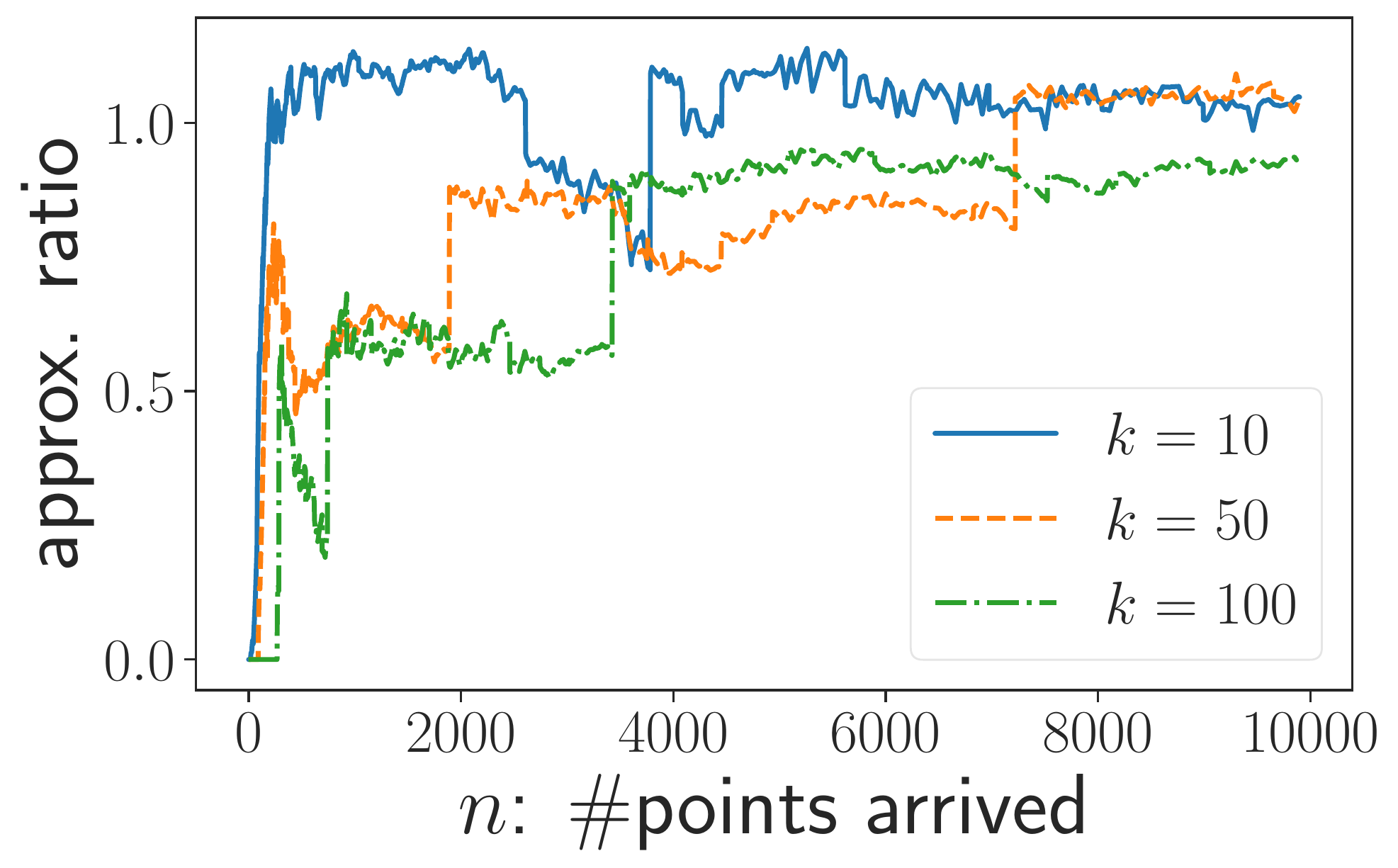}
  }
  \subfloat[\textsc{Covertype}]{
    \includegraphics[width=0.32\textwidth]{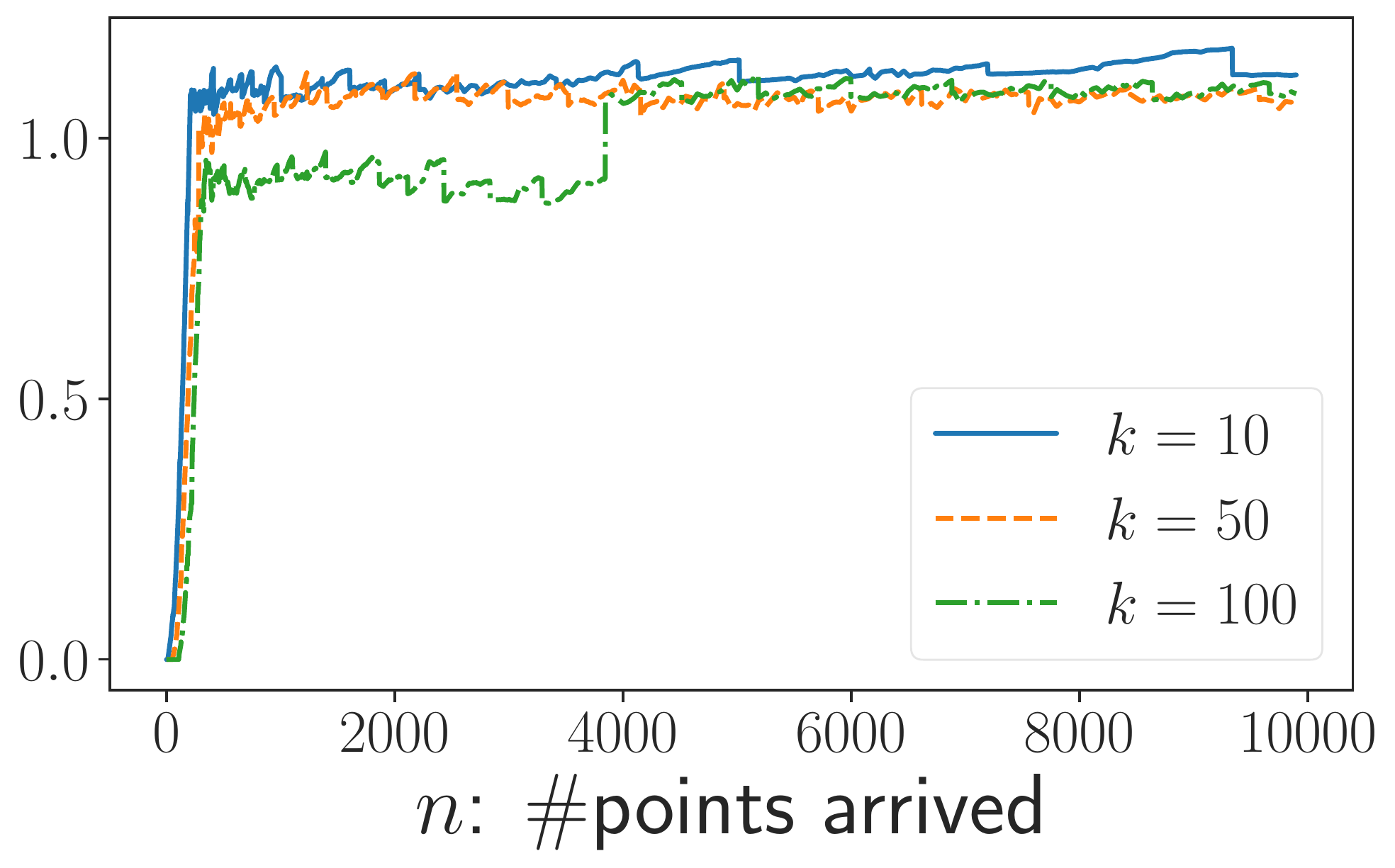}
  }
  \subfloat[\textsc{Letter}]{
    \includegraphics[width=0.32\textwidth]{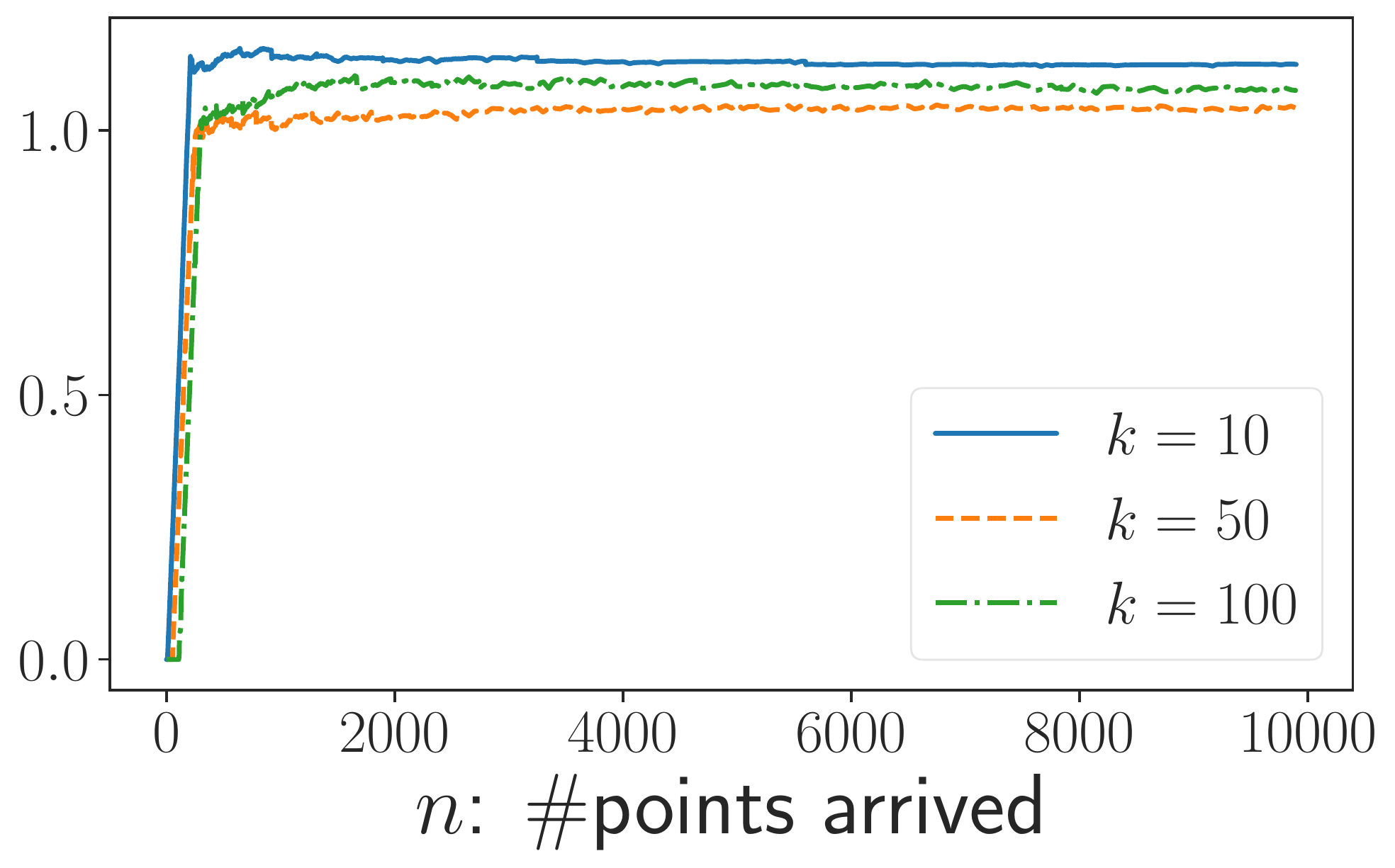}
  }
  \caption{Estimated approximation ratio over time.}
  \label{fig:cost-changing-z}
\end{figure}

\end{document}